\documentclass{article}
\usepackage{graphicx} 
\usepackage{amsmath,amssymb,amsfonts,amscd,amsxtra,amsthm}

\usepackage{latexsym}
\usepackage{xcolor}
\usepackage[title]{appendix}
\usepackage{graphicx}
\usepackage{subcaption} 
\usepackage{booktabs}
\usepackage{bbm}
\usepackage{float}
\usepackage{overpic}
\usepackage{caption}
\usepackage{tikz-cd}
\usepackage{tikz}
\usepackage{algorithmic}
\usepackage{cite}
\usepackage[ruled,linesnumbered,noend]{algorithm2e}

\usepackage{appendix}
\usepackage{epstopdf}
\usepackage[top=1in, bottom=1in, left=1in, right=1in]{geometry}
\renewcommand{\theequation}{\thesection.\arabic{equation}}

\newtheorem{thm}{Theorem}[section]
\newtheorem{cor}[thm]{Corollary}
\newtheorem{lem}[thm]{Lemma}
\newtheorem*{est}{Estimator}
\newtheorem{prop}[thm]{Proposition}

\newtheorem{rmk}[thm]{Remark}

\renewcommand\appendix{\par
  \setcounter{section}{0}
  \setcounter{subsection}{0}
  \setcounter{figure}{0}
  \setcounter{table}{0}
  \renewcommand\thesection{Appendix \Alph{section}}
  \renewcommand\theequation{\Alph{section}.\arabic{equation}}
  \renewcommand\thefigure{\Alph{section}.\arabic{figure}}
  \renewcommand\thetable{\Alph{section}.\arabic{table}}
  \renewcommand\thethm{\Alph{section}.\arabic{thm}}
}

\numberwithin{equation}{section}

\DeclareMathAlphabet{\itbf}{OML}{cmm}{b}{it}
\DeclareMathAlphabet\mathbfcal{OMS}{cmsy}{b}{n}

\renewcommand{\hat}{\widehat}

\def\RR{\mathbb{R}}

\def\EE{\mathbb{E}}

\def\bx{{{\itbf x}}}

\def\by{{{\itbf y}}}
\def\bz{{{\itbf z}}}

\def\bp{{\itbf p}}

\def\bg{{\itbf g}}
\def\be{{\itbf e}}
\def\br{{\itbf r}}

\def\bf{{\itbf f}}

\def\bet{{\boldsymbol{\eta}}}

\def\caP{\mathcal{P}}

\def\cI{{\mathcal I}}
\def\cA{{\mathcal A}}
\def\cB{{\mathcal B}}
\def\cD{{\mathcal D}}

\def\cO{{\mathcal O}}

\def\cS{{\mathcal S}}

\def\cG{{\mathcal G}}

\def\bet{{\boldsymbol{\eta}}}

\def\lb{\left <}

\def\rb{\right >}

\def\cT{\mathcal{T}}
\def\cK{{\mathcal K}}

\def\12{{\frac{1}{2}}}

\newcommand{\bra}[1]{\left(#1\right)}
\usepackage{authblk}

\newcommand{\Rn}[1]{\textup{\uppercase\expandafter{\romannumeral #1}}}

\newcommand{\blue}[1]{#1}
\newcommand{\red}[1]{#1}

\title{Passive Imaging with Ambient Noise Under Wave Speed Mismatch: Mathematical Analysis and Wave Speed Estimation}
\author{
Zetao Fei\thanks{\footnotesize CMAP, Ecole Polytechnique, Institut Polytechnique de Paris, 91120 Palaiseau, France (zetao.fei@polytechnique.edu).}\;
	and Josselin Garnier\thanks{\footnotesize 
	CMAP, Ecole Polytechnique, Institut Polytechnique de Paris, 91120 Palaiseau, France (josselin.garnier@polytechnique.edu).}}

\begin{document}

\maketitle
\begin{abstract}
    It is known that waves generated by ambient noise sources and recorded by passive receivers can be used to image the reflectivities of an unknown medium. However, reconstructing the reflectivity of the medium from partial boundary measurements remains a challenging problem, particularly when the background wave speed is unknown. In this paper, we investigate passive correlation-based imaging in the daylight configuration, where uncontrolled noise sources illuminate the medium and only ambient fields are recorded by a sensor array. We first analyze daylight migration for a point reflector embedded in a homogeneous background. By introducing a searching wave speed into the migration functional, we derive an explicit characterization of the deterministic shift and defocusing effects induced by wave-speed mismatch. We show that the maximum of the envelope of the resulting functional provides a reliable estimator of the true wave speed. We then extend the analysis to a random medium with correlation length smaller than the wavelength. Leveraging the shift formula obtained in the homogeneous case, we introduce a virtual guide star that remains fixed under migration with different searching speeds. This property enables an effective wave-speed estimation strategy based on spatial averaging around the virtual guide star. For both homogeneous and random media, we establish resolution analyses for the proposed wave-speed estimators. Numerical experiments are conducted to validate the theoretical result.
\end{abstract}
\section{Introduction}
Passive correlation--based imaging reconstructs information about a medium or subsurface using only wavefields transmitted by opportunistic or ambient noise sources. This technique was first demonstrated experimentally in \cite{weaver2001ultrasonics} and has since gained widespread use in various fields, including helioseismology \cite{duvall1993time}, volcano monitoring \cite{sabra2006extracting}, and reservoir monitoring \cite{curtis2006seismic}. For a broad interdisciplinary review of correlation-based methods, we refer the reader to \cite{larose2006correlation}; an extensive overview of the mathematical foundations can be found in the monograph \cite{Garnier_Papanicolaou_2016}.

It is now well understood that key information---such as travel times and Green functions---is encoded in the cross-correlation of recorded signals \cite{lobkis2001emergence,PhysRevE.69.046610,de2009semiclassical, Bardos_2008}. This observation underlies the success of imaging reflectors using migration of the empirical cross-correlation. In \cite{doi:10.1137/080723454}, the authors developed a unified mathematical framework for analyzing passive imaging across diverse configurations. The uniqueness of passive imaging inverse problems has also been established in both the time and frequency domains; see, for instance, \cite{HELIN2018132, agaltsov2020global}.

In this paper, we focus on passive imaging in the \emph{daylight configuration}, where a sensor array is located between randomly distributed noise sources and the medium of interest. This setting is common in the literature on virtual-source imaging, where the receiver acts as a virtual source. The scattered waves recorded at the array are numerically back-propagated as if emitted from the receiver, allowing wavefield reconstruction and refocusing at reflector locations. This method is widely used in seismic exploration \cite{bakulin2006virtual, wapenaar2017virtual, https://doi.org/10.1111/1365-2478.12495}, and its mathematical foundations are discussed in \cite{garnier2014role}. We first investigate the case of a point-like reflector embedded in a homogeneous background, a geometry closely related to classical reflector-imaging problems. Cross-correlation based method in this setting has been validated in physical experiments \cite{https://doi.org/10.1111/j.1365-2478.2007.00684.x} and widely applied in seismology, e.g., \cite{harmankaya2013locating, kaslilar2013estimating, konstantaki2013imaging}. A detailed resolution analysis of this approach was provided in \cite{garnier2010resolution,Garnier_2012}. We then extend the proposed method to a second scenario in which the medium is random, with a correlation length much smaller than the wavelength.

A major focus of this work concerns the estimation of the background wave speed from partial passive boundary measurements. Accurate knowledge of the background wave speed is crucial in passive imaging, both for properly focusing the image and for interpreting medium properties. Eikonal-based methods are widely used for wave-speed estimation in ambient noise tomography \cite{shapiro2005high, lin2009eikonal}; see also the review \cite{ritzwoller2011ambient}. However, their implementation in a daylight configuration is challenging. Optimization-based approaches such as Full Waveform Inversion (FWI) \cite{sager2018towards, fichtner_tsai_2019} offer an alternative strategy by minimizing the discrepancy between recorded and simulated data, and recent efforts have extended such methods to quantitative passive imaging \cite{muller2024quantitative}. Nonetheless, these techniques remain computationally demanding and suffer from intrinsic non-convexity. In the literature, cross-correlation–based methods have been applied to image media with three-dimensional rapid fluctuations and a one-dimensional slowly varying background in the parabolic regime \cite{Garnier_2009}. However, to our knowledge, few studies have considered estimating the effective wave speed in a random medium under passive, daylight-type illumination.

In the context of active imaging, recent developments in aberration correction and wave-speed estimation using reflection-matrix analysis have introduced techniques that do not require a strong guide star (i.e. a point-like strong reflector). These methods exploits the recorded backscattered field by generating a synthetic focus that acts as a virtual guide star \cite{lambert2022ultrasound, 9858891, bureau2024reflection}. This idea motivates the introduction of a virtual guide star in our passive imaging setting. A mathematical framework for interpreting such experiments was recently proposed in \cite{Josselin_Garnier_2025_Inverse_Problems_and_Imaging}, and a related concept has been explored in passive imaging in \cite{giraudat2024matrix}.

The contributions of this paper are summarized below. First, in the daylight configuration with a point-like reflector in a homogeneous background, we introduce a searching wave speed $c_s$ into the traditional daylight imaging functional and derive a quantitative characterization of the effects of wave-speed mismatch. Our analysis shows how an incorrect wave speed leads to shifts and blurring in the reconstructed reflector location. We further demonstrate that the maximum of the envelope of the generalized imaging functional yields a robust estimator for the background wave speed. We also identify a special degenerate case within the narrowband regime: under this configuration, accurate reflector localization and background wave-speed estimation cannot be achieved simultaneously, see Section \ref{subsec: imaging func} for a precise statement. Second, we extend this methodology to the random-medium setting, where no strong guide star is present. By migrating the empirical cross-correlation using different searching wave speeds, we construct a virtual guide star whose physical location is stabilized by the analytical shifting formula established in the homogeneous case (see~\eqref{eq: z_f}). The background wave speed is then estimated from spatial averages of the imaging functional around this virtual guide star. For both settings, we provide a comprehensive theoretical analysis, including resolution properties, and validate the proposed estimators through numerical experiments.

The remainder of the paper is organized as follows. Section~\ref{sec: preliminary} introduces the mathematical framework and reviews key properties of cross-correlation in passive imaging. Section~\ref{sec: guide star} analyzes the generalized daylight migration imaging functional in the presence of a point-like reflector and presents a wave-speed estimator. Section~\ref{sec: random wave speed} extends the method to random media and introduces the virtual guide-star--based estimator. Numerical experiments are given in Section~\ref{sec: numerical}, and concluding remarks appear in Section~\ref{sec: conclusion}.

\section{Preliminaries} \label{sec: preliminary}
In this section, we present the mathematical models and assumptions underlying the problem studied in this paper. Section~\ref{subsec: setup} provides a detailed description of the problem setup. In Section~\ref{subsec: cross-cor}, we briefly review key properties of cross-correlation, a fundamental tool in passive imaging. For a more comprehensive discussion, we refer the reader to \cite{Garnier_Papanicolaou_2016}.

\subsection{Problem Outline} \label{subsec: setup}

We consider the scalar wavefield \(u\) governed by the inhomogeneous acoustic wave equation
\begin{align} \label{eq: wave_model}
    \frac{1}{c^2(\bx)} \frac{\partial^2 u}{\partial t^2} - \Delta_{\bx} u = n(t,\bx),
\end{align}
where \(c(\bx)\) denotes the spatially varying wave speed and \(n(t,\bx)\) is a random noise field.  
The noise is assumed to be zero-mean and stationary in time, with frequency band \(\mathcal{B}\) centered at \(\omega_c\) and bandwidth \(B\). Its statistics are characterized by
\begin{align}
    \mathbb{E}\!\left[n(t,\bx)n(t',\bx')\right]
    =
    F_\varepsilon(t-t')\,\Gamma(\bx,\bx'),
\end{align}
where \(F_\varepsilon\) denotes the normalized temporal autocorrelation function. It is even and attains its maximum at the origin, and is normalized so that $F_\varepsilon(0)=1$.

Throughout the paper, we use the Fourier transform convention
\[
\hat f(\omega)
=
\int dt~
f(t)e^{i\omega t}.
\]
We denote by \(\hat F_\varepsilon\) the Fourier transform of \(F_\varepsilon\). By Bochner's theorem, \(\hat F_\varepsilon(\omega)\) is real-valued, even, and nonnegative.

We assume that the decoherence time of the sources is short relative to the typical travel times. Introducing a small scale parameter \(\varepsilon\ll 1\), we write
\begin{align}\label{eq: scaling_F_polished}
    F_\varepsilon(t) = F\!\left(\frac{t}{\varepsilon}\right),
\end{align}
so that the central frequency scales as \(\omega_c = \omega_0/\varepsilon\), and the bandwidth scales as \(B = B_H/\varepsilon\), and the typical wavelength scales as $\lambda = \varepsilon \frac{2\pi c_0}{\omega_0}$.

Spatially, the sources are assumed to be delta-correlated:
\[
    \Gamma(\by_1,\by_2) = K(\by_1)\delta(\by_1-\by_2),
\]
where \(K\) describes their spatial distribution. Extensions to correlated sources are possible, see \cite{Bardos_2008}.\\

\noindent\textbf{Medium model.}
The wave speed is described as a small perturbation of a homogeneous background:
\[
    \frac{1}{c^2(\bx)} = \frac{1}{c_0^2}\left(1+\rho(\bx)\right).
\]
We consider two types of perturbations:

\begin{itemize}

\item \textbf{Point-like reflector.}  
A compact reflector located near \(\bz_r\) is modeled by
\[
    \rho(\bx) = \sigma_r V_r(\bx),
    \qquad 
    V_r(\bx)=\mathbbm{1}_{\Omega_r}(\bx-\bz_r),
\]
where \(\Omega_r\) has characteristic size \(\ell_r\).  
The point-reflector regime corresponds to \(\ell_r\ll \lambda\), the reflector is assumed weak, \(\sigma_r\ll 1\).  
Neither the reflector position \(\bz_r\) nor the background velocity \(c_0\) is assumed to be known.

\item \textbf{Random medium.}  
In this regime, the perturbations are supported in a compact region \(\mathcal{D}\subset \RR^3\) and modeled by a rapidly oscillatory random field:
\[
    \rho(\bx) = \rho_\mu(\bx)
    = \mu\!\left(\frac{\bx}{\ell_c}\right),
    \qquad \bx\in\mathcal{D},
\]
where \(\ell_c\ll\lambda\) is the correlation length.  
The random field \(\mu\) is assumed to be stationary, mean-zero, and has covariance function
\[
    \Sigma(\bx) = \operatorname{Cov}(\mu(\bx),\mu(0)),
\]
decaying at the rate \(|\Sigma(\bx)|\lesssim (1+|\bx|)^{-4-\eta}\) for some \(\eta>0\).  
The fluctuations are assumed weak, \(|\Sigma(0)|\ll 1\).  
While the mean-zero assumption is restrictive theoretically, it is not an obstacle in practice because the boundary reflections due to a nonzero mean can be mitigated or are negligible in the imaging geometry presented below.

\end{itemize}

\noindent\textbf{Data acquisition geometry.}
We work in a passive \emph{daylight} configuration in which an array \(\mathcal{A}\) records the signals \(u(t,\bx_j)\), \(j=1,\dots,N\), and is positioned between the random noise sources (i.e. the support of $K$) and the medium of interest (i.e. $\bz_r$ or $\mathcal{D}$).  
The short decoherence time implies a high central frequency.  
The array radius \(a\) is assumed larger than the typical wavelength \(\lambda\), but smaller than the distance to the medium of interest.

In summary, we consider a multiscale regime governed by a single small parameter $\varepsilon$, with the characteristic propagation distance between the array and the source or imaging region chosen as the reference $\cO(1)$ length scale. Relative to this scale, the wavelength is assumed to be small and scales as $\lambda \sim \cO(\varepsilon)$. The array aperture occupies an intermediate scale and satisfies $ a \sim \cO(\varepsilon^{1/2})$:
\[
    \omega_c = \frac{\omega_0}{\varepsilon}, \quad \lambda = \varepsilon\frac{2\pi c_0}{\omega_0}, \quad a = \varepsilon^{\frac{1}{2}}a_0.
\]
\red{Here, $a_0\sim\mathcal{O}(1)$ is the normalized array radius. We denote the corresponding normalized array by $\mathcal{A}_0$.}
Consequently,
\[
\lambda \ll a \ll 1,
\]
which expresses a clear separation of length scales: the array is large compared with the wavelength while remaining small compared with the propagation distance. The source location $\bz_r$ and the imaging region $\cD$ are both assumed to lie at distances of order one from the array,
\[
\operatorname{dist}(\bz_r,\mathcal A) \sim \cO(1), \qquad \operatorname{dist}(\cD,\mathcal A) \sim \cO(1).
\]
This should be understood as a scale-separation assumption defining the asymptotic regime rather than as a physical restriction on the array geometry. The resulting configuration corresponds to the standard daylight imaging setting encountered in many applications of passive wave imaging, including seismology and related inverse problems.

The primary goal of this paper is to propose estimators for the background wave speed $c_0$ for both configurations and to localize the point-like reflector in the first configuration. Figure~\ref{fig: daylight_config} illustrates the daylight configuration for the point-reflector and random-medium setups.

\begin{figure}[htbp]
    \centering

    \begin{subfigure}{0.45\textwidth}
        \centering
        \begin{tikzpicture}

        \foreach \i in {0,...,24} {
            \pgfmathsetmacro\y{4.0 - 0.1 * \i}
            \pgfmathsetmacro\x{rnd * 0.6 - 0.2}
            \draw (\x, \y) circle (2pt);
        }

        \draw[thick] (3,2.2) rectangle (3.3,3.3);
        \node at (3.5,2) {$\cA$};

        \foreach \i in {1,...,5} {
            \pgfmathsetmacro\y{2.2 + \i * (3.3 - 2.2)/6}
            \draw (3, \y) -- (3.3, \y);
        }

        \fill (5.7,3.4) circle (2pt);
        \node at (6,3.6) {$\bz_r$};
        \end{tikzpicture}
        \caption{Daylight configuration for point-like reflector. \red{The characteristic size of the reflector is much smaller than the wavelength: $\ell_r \ll \lambda$.}}
    \end{subfigure}
    \hspace{0.02\textwidth}
    \begin{subfigure}{0.45\textwidth}
        \centering
        \begin{tikzpicture}

        \foreach \i in {0,...,24} {
            \pgfmathsetmacro\y{4.0 - 0.1 * \i}
            \pgfmathsetmacro\x{rnd * 0.6 - 0.2}
            \draw (\x, \y) circle (2pt);
        }
        
        \draw[thick] (3,2.2) rectangle (3.3,3.3);
        \node at (3.5,2) {$\cA$};
        
        \foreach \i in {1,...,5} {
            \pgfmathsetmacro\y{2.2 + \i * (3.3 - 2.2)/6}
            \draw (3, \y) -- (3.3, \y);
        }
        
        \draw[thick, dashed, rounded corners=5pt] (4.0,1.6) rectangle (6.6,4.0);
        \node at (6.2,3.6) {$\cD$};
        \draw[thick]
            (4.8,2.9) --  
            (4.95,3.2) --  
            (5.2,2.8) --  
            (5.35,3.5) --  
            (5.5,2.7) --  
            (5.65,3.3) --  
            (5.8,3.0);    
        \draw[<->] (5.2,2.4) -- (5.5,2.4) node[midway, below] {$\ell_c$};
        
        \end{tikzpicture}

        \caption{Daylight configuration for random medium. \red{The correlation length of the random medium is much smaller than the wavelength: $\ell_c \ll \lambda$.}}
    \end{subfigure}

    \caption{Daylight configuration for passive imaging of two different types of medium. \red{The circles stand for the noise sources whose support is described by the function $K(\bx)$.}}
    \label{fig: daylight_config}
\end{figure}

\subsection{Cross-correlation and imaging functional}\label{subsec: cross-cor}

We start the introduction to the cross-correlation with the representation of the wave function (\ref{eq: wave_model}). By the formulation provided previously, the solution to the wave equation has the following integral representation:
\begin{align}
    u(t,\bx) = \iint ds d\by\ n(t-s,\by)G(s,\bx,\by).
\end{align}
The function $G$ is the time-dependent causal Green's function, i.e., the fundamental solution to the equation 
\begin{align}
    \left\{
    \begin{array}{l}
        \frac{1}{c^2(\bx)} \frac{\partial^2 G}{\partial t^2} - \Delta_{\bx} G=\delta(t)\delta(\bx-\by)\\
        G(t,\bx,\by) =0, \quad \forall t<0.
    \end{array}
    \right.
\end{align}
We also introduce the time-harmonic Green's function $\hat{G}_0$ for the background homogeneous medium, i.e., the solution to the Helmholtz equation 
\begin{align}
    \left\{
    \begin{array}{l}
        \Delta_\bx \hat{G}_0 + \frac{\omega^2}{c^2_0} \hat{G}_0 = -\delta(\bx - \by),\\
        \lim_{|\bx| \to \infty} |\bx| \left( \frac{\bx}{|\bx|} \cdot \nabla_{\bx} - i \frac{\omega}{c_0} \right) \hat{G}_0(\omega, \bx, \by) = 0.
    \end{array}
    \right.
\end{align}

Explicitly, we write the time-harmonic Green's function in three dimensions as
\begin{align}
    \hat{G}_0(\omega, \bx, \by) = \frac{1}{4 \pi |\bx - \by|} \exp\bra{i \frac{\omega}{c_0} |\bx - \by|}.
\end{align}

We consider the empirical cross-correlation over a time interval $(0,T)$ of the signal recorded at $\bx_1, \bx_2 \in \cA$ defined by 
\begin{align}
    C_T(\tau,\bx_1,\bx_2) = \frac{1}{T}\int_0^T dt~u(t,\bx_1)u(t+\tau,\bx_2).
\end{align} 
The statistical stability of the above quantity is well established, as shown in the following proposition \cite{Garnier_Papanicolaou_2016}:
\begin{prop}
    \begin{itemize}
    \item[(1)] The expectation of the empirical cross-correlation \( C_T \) (with respect to the statistical law of the sources) is independent of \( T \):
    \begin{align}
        \EE C_T(\tau, \bx_1, \bx_2) = C^{(1)}(\tau, \bx_1, \bx_2).
    \end{align}

    The theoretical (or statistical) cross-correlation \( C^{(1)} \) is given by
    \begin{align}\label{eq: theo cross-cor}
        C^{(1)}(\tau, \bx_1, \bx_2) 
        = \frac{1}{2\pi} \iint d\by \, d\omega \, \hat{F}(\omega) K(\by)\overline{\hat{G}\left(\omega, \bx_1, \by \right)} 
        \hat{G}\left(\omega, \bx_2, \by\right) e^{-i\omega \tau},
    \end{align}
    where \( \hat{G}(\omega, \bx, \by) \) is the time-harmonic Green's function, that is, the Fourier transform of \( G(t, \bx, \by) \), and the overline stands for complex conjugation.

    \item[(2)] The empirical cross-correlation \( C_T \) is a self-averaging quantity:
    \begin{equation}
        C_T(\tau, \bx_1, \bx_2) \stackrel{T\rightarrow{\infty}}{\longrightarrow} C^{(1)}(\tau, \bx_1, \bx_2),
    \end{equation}
    in probability with respect to the statistical law of the sources.
\end{itemize}
\end{prop}

With the recorded signals $\{u(t,\bx_j)\}_{t\in\RR,\, j=1,\dots,N}$, we compute the corresponding cross-correlations. The inverse problem considered in this paper is then reformulated as extracting information about the medium from the data set $\{C_T\}$. To this end, we introduce the \textit{generalized daylight migration imaging functional}
\begin{align}\label{eq: imaging functional}
    \cI(\bz_s,c_s) = \sum_{j,l=1}^N  
    C_T\!\bra{\cT_s(\bz_s,\bx_l)+\cT_s(\bz_s,\bx_j),\bx_j,\bx_l},
\end{align}
where $\bz_s$ denotes the \red{\red{search point}}, $c_s$ is the sound speed used in migration, and $\cT_s(\bz_s,\bx)$ is the generalized travel time defined by
\begin{align}\label{eq: def_time}
\cT_s(\bz_s,\bx):= \frac{|\bz_s-\bx|}{c_s}.
\end{align}
Furthermore, we denote the ratio of $c_s$ and $c_0$ by $v = c_s/c_0$. When $c_s=c_0$ (i.e. $v=1$), $\cT_s$ coincides with the travel time in the homogeneous medium,
\[
\cT_0(\bz_s,\bx):= \frac{|\bz_s-\bx|}{c_0},
\]
and the imaging functional reduces to the classical daylight migration functional.

In the dense-array limit, the imaging functional can be approximated by
\begin{align}
    \cI(\bz_s,c_s)
    = \frac{N^2}{|\cA|^2}
    \iint_{\cA^2} d\bx_1 d\bx_2\,
    C_T\!\bra{\cT_s(\bz_s,\bx_2)+\cT_s(\bz_s,\bx_1),\bx_1,\bx_2}.
\end{align}

The construction of the imaging functional motivates the analysis of the theoretical cross-correlation $C^{(1)}$. Under the assumption that the perturbation of the background medium is weak, we may invoke the Born approximation for the Green’s function,
\begin{align}\label{eq: born's approx}
    \hat{G}(\omega, \bx, \by)
    \simeq \hat{G}_0(\omega, \bx, \by)
    + \frac{\omega^2}{c_0^2}
    \int d\bz\, \widetilde{\rho}(\bz)
    \hat{G}_0(\omega, \bx, \bz)
    \hat{G}_0(\omega, \bz, \by),
\end{align}
where
\[
    \widetilde{\rho}(\bz)=
    \begin{cases}
        \sigma_r \ell_r^3 \delta(\bz-\bz_r), & \text{point-like reflector},\\
        \rho_\mu(\bz), & \text{random medium}.
    \end{cases}
\]

Combining~\eqref{eq: theo cross-cor}, \eqref{eq: born's approx}, and the scaling regime introduced in the previous section, we obtain the decomposition
\begin{align}
     C^{(1)}(\tau,\bx_1,\bx_2)
     \simeq C^{(1)}_{0}(\tau,\bx_1,\bx_2)
     + C^{(1)}_{\text{\Rn{1}}}(\tau,\bx_1,\bx_2)
     + C^{(1)}_{\text{\Rn{2}}}(\tau,\bx_1,\bx_2),
\end{align}
where
\begin{align*}
    & C^{(1)}_{0}(\tau,\bx_1,\bx_2)
    = \frac{1}{2\pi}
    \iint d\by d\omega\,
    K(\by)\hat{F}(\omega)
    \overline{\hat{G}_0}\!\bra{\tfrac{\omega}{\varepsilon},\bx_1,\by}
    \hat{G}_0\!\bra{\tfrac{\omega}{\varepsilon},\bx_2,\by}
    e^{-i\frac{\omega\tau}{\varepsilon}},\\[0.3em]
    & C^{(1)}_{\text{\Rn{1}}}(\tau,\bx_1,\bx_2)
    = \frac{1}{2\pi c_0^2\varepsilon^2}
    \int d\bz\, \widetilde{\rho}(\bz)
    \iint d\by d\omega\,
    \omega^2 K(\by)\hat{F}(\omega)
    \hat{G}_0\!\bra{\tfrac{\omega}{\varepsilon},\bx_2,\bz}
    \hat{G}_0\!\bra{\tfrac{\omega}{\varepsilon},\bz,\by}
    \overline{\hat{G}_0}\!\bra{\tfrac{\omega}{\varepsilon},\bx_1,\by}
    e^{-i\frac{\omega\tau}{\varepsilon}},\\[0.3em]
    & C^{(1)}_{\text{\Rn{2}}}(\tau,\bx_1,\bx_2)
    = \frac{1}{2\pi c_0^2\varepsilon^2}
    \int d\bz\, \widetilde{\rho}(\bz)
    \iint d\by d\omega\,
    \omega^2 K(\by)\hat{F}(\omega)
    \overline{\hat{G}_0}\!\bra{\tfrac{\omega}{\varepsilon},\bx_1,\bz}
    \overline{\hat{G}_0}\!\bra{\tfrac{\omega}{\varepsilon},\bz,\by}
    \hat{G}_0\!\bra{\tfrac{\omega}{\varepsilon},\bx_2,\by}
    e^{-i\frac{\omega\tau}{\varepsilon}}.
\end{align*}

The term \(C^{(1)}_{0}\) corresponds to the contribution of the direct waves, and it involves the rapid phase
\[
\frac{\omega}{\varepsilon} \Bigl( \frac{|\bx_2-\by|}{c_0} - \frac{|\bx_1-\by|}{c_0} - \tau \Bigr).
\]
For fixed \(\bx_1,\bx_2 \in \cA\), this contribution is supported on the time interval \([-\Delta \tau_0, \Delta \tau_0]\), where
\[
\Delta \tau_0
:= \frac{\sup_{\by \in \operatorname{supp}(K)} \bigl|\,|\by-\bx_1| - |\by-\bx_2|\,\bigr|}{c_0}
\le \frac{|\bx_1-\bx_2|}{c_0}
\le \frac{\operatorname{diam}(\cA)}{c_0},
\]
which is of order \(\mathcal{O}(\varepsilon^{\frac{1}{2}})\) in the regime under consideration.

In contrast, the scattered-wave contributions \(C^{(1)}_{\text{\Rn{1}}}\) and
\(C^{(1)}_{\text{\Rn{2}}}\) involve the rapid phases
\[
\frac{\omega}{\varepsilon}
\Bigl(
\frac{|\bx_2-\bz|}{c_0}
+ \frac{|\bz-\by|}{c_0}
- \frac{|\bx_1-\by|}{c_0}
- \tau
\Bigr)
\quad\text{and}\quad
\frac{\omega}{\varepsilon}
\Bigl(
\frac{|\bx_2-\by|}{c_0}
- \frac{|\bz-\by|}{c_0}
- \frac{|\bx_1-\bz|}{c_0}
- \tau
\Bigr),
\]
respectively. Recall that we consider the daylight configuration in this paper, these terms are supported on time intervals satisfying
\[
|\tau| \asymp \frac{2\,\operatorname{dist}(\bz_r,\cA)}{c_0}
\quad \text{or} \quad
|\tau| \asymp \frac{2\,\operatorname{dist}(\cD,\cA)}{c_0},
\]
both of which are quantities of order one. Consequently, within the considered regime, the direct and scattered contributions to the cross-correlation are well separated in time. Therefore, the component $C^{(1)}_{0}$ can be effectively removed from the recorded data by time-windowing.

In the remainder of the analysis, we focus exclusively on the scattered-wave contributions. With a slight abuse of notation, we therefore redefine
\begin{align}
    C^{(1)}(\tau,\bx_1,\bx_2)
    = C^{(1)}_{\text{\Rn{1}}}(\tau,\bx_1,\bx_2)
    + C^{(1)}_{\text{\Rn{2}}}(\tau,\bx_1,\bx_2).
\end{align}

\section{Guide star assisted wave speed estimation}\label{sec: guide star}
In this section, we analyze the properties of the imaging functional $\cI(\bz_s,c_s)$ defined in \eqref{eq: imaging functional} when a point-like reflector is present in the medium. The analysis quantitatively characterizes the influence of wave speed mismatch on the point spread function, including the phenomenon of focusing point shift. Based on these insights, we propose methods for estimating the background wave speed $c_0$ and localizing the point-like reflector. The imaging functional considered here generalizes that of \cite{Garnier_Papanicolaou_2016}. Accordingly, the present analysis may be viewed as an extension of the framework developed therein to the problem of background wave-speed estimation.

\subsection{Analysis of the imaging functional}\label{sec: imaging functional}\label{subsec: imaging func}
To present the properties of the imaging functional, we first introduce a coordinate frame $\bra{\hat{\be}_1,\hat{\be}_2,\hat{\be}_3}$ satisfying 
\begin{itemize}
    \item the sensor array $\cA\subset\operatorname{span}\{\hat{\be}_1,\hat{\be}_2\}$,
    \item the reflector $\bz_r\in\operatorname{span}\{\hat{\be}_1,\hat{\be}_3\}$ with coordinates $\bz_r = (z_{r,1},0,z_{r,3})$.
\end{itemize}

\blue{Throughout the analysis, we adopt the vector decomposition that $\boldsymbol{\zeta} = (\boldsymbol{\zeta}^\perp, \zeta^{\parallel})$, where $\boldsymbol{\zeta}^\perp = (\zeta_1,\zeta_2)$ and $\zeta^{\parallel} = \zeta_3$.} Recall the regime we adopt for the sensor array, for $\bx_1\in\cA$, we write $\bx_1 = \varepsilon^{\frac{1}{2}}\widetilde{\bx_1} = \varepsilon^{\frac{1}{2}}\bra{\widetilde{x_{1,1}},\widetilde{x_{1,2}},0}$.

We also introduce the following orthonormal frame $\{\hat{\bf}_1,\hat{\bf}_2,\hat{\bf}_3\}$:
\begin{align}
    \hat{\bf}_1 = \frac{1}{|\bz_r|}\bra{z_{r,3}\hat{\be}_1-z_{r,1}\hat{\be}_3}, \quad \hat{\bf}_2 = \hat{\be}_2, \quad \hat{\bf}_3 = \frac{1}{|\bz_r|}\bra{z_{r,1}\hat{\be}_1+z_{r,3}\hat{\be}_3}.
\end{align}

Different from the traditional daylight migration imaging functional, in $\cI(\bz_s,c_s)$, the speed of sound used for backpropagation $c_s$ can have a mismatch with $c_0$. This mismatch can lead to a shift in the focusing point of the image derived from $\cI(\bz_s,c_s)$. A similar phenomenon is observed in the active imaging configuration, see \cite{Josselin_Garnier_2025_Inverse_Problems_and_Imaging}. We denote the focusing point in the image as $\bz_f(v)$, and $\bz_f$ when there is no ambiguity. We define the map $\phi_v:\RR^3 \to \RR^3$ as 
\begin{align}\label{eq: phi_mapping}
    \phi_v(\bz_r) = \bra{v^2\bz_r^\perp,v\sqrt{(z_r^{\parallel})^2+(1-v^2)|\bz_r^\perp|^2}}.
\end{align}
\blue{
In Theorem \ref{thm: functional}, we show that when the point-like reflector is positioned at $\bz_r$, under certain conditions, the focusing point is located at 
\begin{align}
    \bz_f = \phi_v(\bz_r), \quad v = \frac{c_s}{c_0}.
\end{align}
}
Before diving into the mathematical calculation, we first introduce an orthonormal frame $\{\hat{\bg}_1,\hat{\bg}_2,\hat{\bg}_3\}$ based on $\bz_f$ as follows:
\begin{align}
    \hat{\bg}_1 = \frac{1}{|\bz_f|}\bra{z_{f,3}\hat{\be}_1-z_{f,1}\hat{\be}_3}, \quad \hat{\bg}_2 = \hat{\be}_2, \quad \hat{\bg}_3 = \frac{1}{|\bz_f|}\bra{z_{f,1}\hat{\be}_1+z_{f,3}\hat{\be}_3}.
\end{align}
We illustrate the three coordinate systems in Figure \ref{fig: coordinates}.
\begin{figure}[!ht]
    \centering
    \begin{tikzpicture}

    \draw[thick] (3,2.2) rectangle (3.3,3.3); 
    \node at (3.5,2) {$\cA$};               
 
    \foreach \i in {1,...,5} {
        \pgfmathsetmacro\y{2.2 + \i * (3.3 - 2.2)/6}
        \draw (3, \y) -- (3.3, \y);
    }

    \fill (6,3.4) circle (2pt);
    \node at (6.35,3.1) {$\bz_r$};

    \fill (3.15,2.75) circle (2pt);
    \draw[->,thick] (3.15,2.75) -- (3.15,3.75) node[above] {$\hat{\be}_1$};
    \draw[->,thick] (3.15,2.75) -- (4.15,2.75) node[right] {$\hat{\be}_3$}; 
    \node at (2.8,2.6) {$\mathbf{0}$};

    \draw[dashed,thick] (3.15,2.75) -- (6,3.4);

    \draw[->,thick] (6,3.4) -- (5.77,4.4) node[above] {$\hat{\bf}_1$};
    \draw[->,thick] (6,3.4) -- (7,3.63) node[right] {$\hat{\bf}_3$};
    
    \fill[red] (6.6,3.85) circle (2pt);
    \node[red] at (6.2,4) {$\bz_f$};
    \draw[dashed,thick,red] (3.15,2.75) -- (6.6,3.85);

    \draw[->,thick,red] (6.6,3.85) -- (6.25,4.8) node[above] {\textcolor{red}{$\hat{\bg}_1$}};
    \draw[->,thick,red] (6.6,3.85) -- (7.55,4.16) node[right] {\textcolor{red}{$\hat{\bg}_3$}};

    \end{tikzpicture}
    \caption{The sensor array and three reference frames $(\hat{\be}_1,\hat{\be}_3)$, $(\hat{\bf}_1,\hat{\bf}_3)$, and $(\hat{\bg}_1,\hat{\bg}_3)$. The unit vector $\hat{\be}_2=\hat{\bf}_2=\hat{\bg}_2$ are orthogonal to the plane of the figure.}
    \label{fig: coordinates}
\end{figure}

In the case that $c_s=c_0$, it is known that the resolution of the image depends on the tilt of the array relative to $\hat{\bf}_3$, see \cite{Garnier_Papanicolaou_2016}. In the analysis, we observe that the imaging function also depends on the tilt of the array relative to $\hat{\bg}_3$. Therefore, for the convenience of presentation, we introduce the cosine of the tilt angles:
\begin{align}\label{eq: alpha_def}
    \alpha_r = \lb\hat{\bf}_3,\hat{\be}_3\rb = \frac{z_{r,3}}{|\bz_r|}, \quad \alpha_f = \lb\hat{\bg}_3,\hat{\be}_3\rb = \frac{z_{f,3}}{|\bz_f|}.
\end{align}
We define the normalized point spread function as 
\begin{align}
    \cG_{\alpha_r,\alpha_f}(\xi_1,\xi_2,\xi_3,\xi_4) = \frac{1}{|\cA_0|} \int_{\cA_0} d u_1 du_2 ~\exp\bra{-i(\alpha_fu_1\xi_1+u_2\xi_2)-i\frac{u_1^2\alpha_r^2+u_2^2}{2}\xi_3\red{+}i\frac{u_1^2+u_2^2}{2}\xi_4 }.
\end{align}
In the theoretical results of the imaging functional, $\{\xi_j\}$ correspond to different physical meanings. $\xi_1$ and $\xi_2$ is related to the transverse direction, $\xi_3$ corresponds to the range direction and the influence of the wave speed mismatch is reflected through $\xi_4$. We observe that the function $\cG_{\alpha_r,\alpha_f}$ is a generalization of the point spread function $\cG_{\alpha_r}$ defined in Section 6, \cite{Garnier_Papanicolaou_2016}. The two functions have similar qualitative behaviors when $\xi_3 =0$ (or $\xi_4=0$). When $\xi_3,\xi_4\ne0$, the behavior of the function $\cG_{\alpha_r,\alpha_f}$ is complicated, which leads to the complicated coupling phenomenon of the range component and the mismatch of the wave speed in the point spread function, which is discussed in \red{Section \ref{prop: functional}}.

In the following, we show the asymptotic expression for the generalized imaging function under the regime we consider. Through stationary phase analysis, we quantify the dependence of the imaging functional on the following factors:
\begin{itemize}
    \item the normalized array radius $a_0$,
    \item the distance of the reflector to the center of the sensor array $|\bz_r|$,
    \item the central frequency $\omega_c = \frac{\omega_0}{\varepsilon}$,
    \item the effective bandwidth \blue{$B = \frac{B_H}{\varepsilon}$},
    \item the chosen speed of sound for migration $c_s$.
\end{itemize}
\blue{To properly describe the behavior of the imaging functional near its focusing region, it is useful to introduce a resolution-scale parameterization of the \red{search point}. This parameterization is motivated by the expected anisotropic resolution of the imaging problem. In particular, the cross-range resolution is governed by the Rayleigh criterion and scales as $\lambda L/a$, whereas the range resolution is governed by the signal bandwidth and scales as $c_0/B$. With the parameterization, we state the main results.}
\begin{thm}\label{thm: functional}
    Let us assume that $\hat{F}(\omega)$ has the form 
    \begin{align}\label{eq: profile F_1}
        \hat{F}(\omega) = \frac{1}{B_H}\bra{\hat{F}_{0}\bra{\frac{\omega_0-\omega}{B_H}}+\hat{F}_{0}\bra{\frac{\omega_0+\omega}{B_H}} },
    \end{align}
where $\omega_0, B_H>0$. The generalized daylight imaging functional is non-negligible (i.e. it presents a well-defined peak) when the searching wave speed $c_s$ satisfies \blue{
\begin{align}\label{eq: cond_v}
    v = \frac{c_s}{c_0} \in \left\{
    \begin{array}{ll}
        \left(0, \sqrt{1+\bra{\frac{z_{r,3}}{|\bz_{r}^\perp|}}^2}\right], & \text{if } \bz_{r}^\perp\ne0,\\
        (0,\infty), & \text{if } \bz_{r}^\perp  =0.
    \end{array}
    \right.
\end{align}}
For a given $c_s$ that satisfying (\ref{eq: cond_v}), as a function of $\bz_s$, the imaging functional $\cI(\bz_s,c_s)$ defined by (\ref{eq: imaging functional}) has the focusing point
\begin{align}\label{eq: z_f}
    \bz_f = \phi_v(\bz_r) = \bra{v^2\bz_r^\perp,v\sqrt{\bra{z_r^{\parallel}}^2+(1-v^2)|\bz_r^\perp|^2}}.
\end{align}
At the \red{search point} $\bz_s$:
\begin{align}
    \bz_s = \bz_f+\br_s = \bz_f + \varepsilon^{\frac{1}{2}}v^2 \eta_1 \hat{\bg}_1+\varepsilon^{\frac{1}{2}}v^2 \eta_2 \hat{\bg}_2 +\varepsilon v\eta_3\hat{\bg}_3,
\end{align}
the imaging functional has the asymptotic form (as $\varepsilon\rightarrow 0$):
\begin{align}\label{eq: functional result}
\cI(\bz_s,c_s) \approx & \frac{\sigma_r\ell_r^3N^2}{2^6\pi^3c_0\varepsilon}
\frac{\cK(\mathbf{0},\bz_r)}{|\bz_r|^2} \cdot \caP(\br_s,c_s),
\end{align}
where \blue{
\begin{align}
    \cK(\bx,\bz) = \int_{0}^{\infty} dl~K\bra{\bx+\frac{\bx-\bz}{|\bx-\bz|}  l} ,
\end{align} }
and $\caP(\br_s,c_s)$ is the point spread function defined by 
\begin{align}
   \caP(\br_s,c_s) =  \operatorname{Re} \left\{
\int_\cB d\omega ~ \bra{-i\omega\hat{F}(\omega)}\exp\bra{-i\frac{\omega}{c_0}\bra{2\eta_3+\frac{c_s^2(\eta_1^2+\eta_2^2)}{c_0^2|\bz_r|}}}\right.\notag \\
\times \left.\cG^2_{\alpha_r,\alpha_f}\bra{\frac{-a_0\omega}{c_0|\bz_r|}\eta_1,
\frac{-a_0\omega}{c_0|\bz_r|}\eta_2,0,
\frac{-a_0^2\omega}{c_0|\bz_r|}\bra{\bra{\frac{c_0}{c_s}}^2-1}} \right\}.
\end{align}
If additionally $B_H \ll \omega_0$, then the point spread function has the form 
\begin{align}
    \caP(\br_s,c_s) \simeq&  4\pi\omega_0 \operatorname{Re} \left\{
 -i\exp\bra{-i\frac{\omega_0}{c_0}\bra{2\eta_3+\frac{c_s^2(\eta_1^2+\eta_2^2)}{c_0^2|\bz_r|}}}F_{0}\bra{-\frac{B_H}{c_0}\bra{2\eta_3+\frac{c_s^2(\eta_1^2+\eta_2^2)}{c_0^2|\bz_r|}}}\right.\notag \\
&\times \left.\cG^2_{\alpha_r,\alpha_f}\bra{\frac{-a_0\omega_0}{c_0|\bz_r|}\eta_1,
\frac{-a_0\omega_0}{c_0|\bz_r|}\eta_2,0,
\frac{-a_0^2\omega_0}{c_0|\bz_r|}\bra{\bra{\frac{c_0}{c_s}}^2-1}} \right\},
\end{align}
whose slowly varying envelope is given by
\begin{align}\label{eq: envelope wide}
    \caP_E(\br_s,c_s) \simeq& 2\pi \omega_0\left| F_{0}\bra{-\frac{B_H}{c_0}\bra{2\eta_3+\frac{c_s^2(\eta_1^2+\eta_2^2)}{c_0^2|\bz_r|}}}\right|\notag\\
    &\times\left|\cG_{\alpha_r,\alpha_f}\bra{\frac{-a_0\omega_0}{c_0|\bz_r|}\eta_1,
\frac{-a_0\omega_0}{c_0|\bz_r|}\eta_2,0,
\frac{-a_0^2\omega_0}{c_0|\bz_r|}\bra{\bra{\frac{c_0}{c_s}}^2-1}}\right|^2.
\end{align}
\end{thm}
\begin{proof}
    See \ref{append_3.1}.
\end{proof}
\blue{
In the theorem above, we use the notion of the \emph{slowly varying envelope}. Namely, if a real-valued function \(p(t)\) admits the representation
\[
p(t)=e^{-i\omega_0 t}p_B(t)+\mathrm{c.c.},
\]
where the bandwidth of \(p_B\) is much smaller than the carrier frequency \(\omega_0\), then the slowly varying envelope of \(p(t)\) is defined as \(|p_B(t)|\). We refer the reader to Section~13.4 of \cite{Garnier_Papanicolaou_2016} for a more detailed discussion.}

\blue{
We note that the assumption (\ref{eq: profile F_1}) indicates the noise sources have a central frequency $\omega_0/\varepsilon$ and an effective bandwidth $B_H/\varepsilon$. In this case, the point spread function has a simple form as shown in (\ref{eq: functional result}).
Here, we emphasize that the expression above is valid when the bandwidth of the noise sources is not very small. When the bandwidth is of order $\cO(1)$ \red{(which means it is much smaller than the carrier frequency)}, the expression of the imaging functional is slightly different from the previous one.}
\begin{prop}\label{prop: functional}
    Let us assume that $\hat{F}(\omega)$ has the form 
    \begin{align}\label{eq: profile F_2}
        \hat{F}(\omega) = \frac{1}{\varepsilon B_0}\bra{\hat{F}_{0}\bra{\frac{\omega_0-\omega}{\varepsilon B_0}}+\hat{F}_{0}\bra{\frac{\omega_0+\omega}{\varepsilon B_0}} },
    \end{align}
where $\omega_0, B_0>0$. Under the condition (\ref{eq: cond_v}), 
at the \red{search point} $\bz_s$:
\begin{align}
    \bz_s = \bz_f+\br_s = \bz_f + \varepsilon^{\frac{1}{2}}v^2 \eta_1 \hat{\bg}_1+\varepsilon^{\frac{1}{2}}v^2 \eta_2 \hat{\bg}_2 + v\eta_3\hat{\bg}_3,
\end{align}
such that $|\br_s|\ll|\bz_r|$, the imaging functional has the asymptotic form (as $\varepsilon\rightarrow 0$):
\begin{align}\label{eq: functional result narrow}
\cI(\bz_s,c_s) \approx & \frac{\sigma_r\ell_r^3N^2}{2^6\pi^3c_0\varepsilon}
\frac{\cK(\mathbf{0},\bz_r)}{|\bz_r|^2} \cdot \widetilde \caP(\br_s,c_s),
\end{align}
where $\widetilde \caP(\br_s,c_s)$ is the point spread function defined by 
\begin{align}\label{eq: functional narrow}
    \widetilde \caP(\br_s,c_s)\blue{\simeq}&4\pi\omega_0\operatorname{Re}\left\{ -i\exp\bra{-i\frac{\omega_0}{c_0}\bra{2\frac{\eta_3}{\varepsilon}+\frac{c_s^2(\eta_1^2+\eta_2^2)}{c_0^2|\bz_r|}}} F_{0}\bra{-\frac{2B_0\eta_3}{c_0}}\right. \notag\\
    &\left. \times \cG^2_{\alpha_r,\alpha_f}\bra{\frac{-a_0\omega_0}{c_0|\bz_r|}\eta_1,
\frac{-a_0\omega_0}{c_0|\bz_r|}\eta_2,\frac{-a_0^2\omega_0}{c_0|\bz_r|^2}\eta_3,
\frac{-a_0^2\omega_0}{c_0|\bz_r|}\bra{\bra{\frac{c_0}{c_s}}^2-1}}\right\}.
\end{align}
If $B_0$ is \red{much} larger than the critical value \blue{
\begin{align}
    B_c := \frac{\omega_0}{2}\frac{a_0^2}{|\bz_r|^2} =  \frac{\omega_c}{2}\frac{a^2}{|\bz_r|^2} ,
\end{align}}
then the point spread function has the form 
\begin{align}\label{eq: functional narrow larger}
    \widetilde \caP(\br_s,c_s)\blue{\simeq}&4\pi\omega_0\operatorname{Re}\left\{ -i\exp\bra{-i\frac{\omega_0}{c_0}\bra{2\frac{\eta_3}{\varepsilon}+\frac{c_s^2(\eta_1^2+\eta_2^2)}{c_0^2|\bz_r|}}} F_{0}\bra{-\frac{2B_0\eta_3}{c_0}}\right. \notag\\
    &\left. \times \cG^2_{\alpha_r,\alpha_f}\bra{\frac{-a_0\omega_0}{c_0|\bz_r|}\eta_1,
\frac{-a_0\omega_0}{c_0|\bz_r|}\eta_2,0,
\frac{-a_0^2\omega_0}{c_0|\bz_r|}\bra{\bra{\frac{c_0}{c_s}}^2-1}}\right\},
\end{align}
whose slowly varying envelope is given by 
\begin{align}{\label{eq: envelope wide narrow}}
    \blue{\widetilde{\caP}_E(\br_s,c_s) \simeq} 2\pi\omega_0\left|F_{0}\bra{-\frac{2B_0\eta_3}{c_0}} \right| \cdot\left| \cG_{\alpha_r,\alpha_f}\bra{\frac{-a_0\omega_0}{c_0|\bz_r|}\eta_1,
\frac{-a_0\omega_0}{c_0|\bz_r|}\eta_2,0,
\frac{-a_0^2\omega_0}{c_0|\bz_r|}\bra{\bra{\frac{c_0}{c_s}}^2-1}}\right|^2.
\end{align}
If $B_0$ is \red{much} smaller than the critical value $B_c$, then the point spread function has the form 
\begin{align}\label{eq: functional narrow smaller}
        \widetilde \caP(\br_s,c_s)\blue{\simeq}&4\pi\omega_0\operatorname{Re}\left\{ -i\exp\bra{-i\frac{\omega_0}{c_0}\bra{2\frac{\eta_3}{\varepsilon}+\frac{c_s^2(\eta_1^2+\eta_2^2)}{c_0^2|\bz_r|}}} F_{0}\bra{0}\right. \notag\\
    &\left. \times \cG^2_{\alpha_r,\alpha_f}\bra{\frac{-a_0\omega_0}{c_0|\bz_r|}\eta_1,
\frac{-a_0\omega_0}{c_0|\bz_r|}\eta_2,\frac{-a_0^2\omega_0}{c_0|\bz_r|^2}\eta_3,
\frac{-a_0^2\omega_0}{c_0|\bz_r|}\bra{\bra{\frac{c_0}{c_s}}^2-1}}\right\},
\end{align}
whose slowly varying envelope is given by 
\begin{align}{\label{eq: envelope narrow}}
    \blue{\widetilde{\caP}_E(\br_s,c_s)\simeq} 2\pi\omega_0\left|F_{0}\bra{0} \right| \cdot\left| \cG_{\alpha_r,\alpha_f}\bra{\frac{-a_0\omega_0}{c_0|\bz_r|}\eta_1,
\frac{-a_0\omega_0}{c_0|\bz_r|}\eta_2,\frac{-a_0^2\omega_0}{c_0|\bz_r|^2}\eta_3,
\frac{-a_0^2\omega_0}{c_0|\bz_r|}\bra{\bra{\frac{c_0}{c_s}}^2-1}}\right|^2.
\end{align}
\begin{proof}
    See \ref{append_3.2}.
\end{proof}


\end{prop}
\blue{We note that the assumption (\ref{eq: profile F_2}) indicates the noise sources have a central frequency $\omega_0/\varepsilon$ and an effective bandwidth $B_0$. \red{The critical value $B_c$ characterizes the transition between bandwidth-dominated and aperture-dominated range resolution.} 

Here, we discuss and interpret the results derived in the theorem and proposition above. In the special case $c_s=c_0$, the obtained expressions recover the classical resolution results for daylight passive imaging. Specifically, we analyze both the cross-range and range resolutions, defined as the widths of the point spread function (PSF) in the transverse and longitudinal directions, respectively. We refer to a more detailed introduction to \cite{Garnier_Papanicolaou_2016}. For completness, we summarize the results \blue{in physical parameters} as follows: }
\begin{itemize}
    \item the frames $\{\hat{\bf}_1,\hat{\bf}_2,\hat{\bf}_3\}$ and $\{\hat{\bg}_1,\hat{\bg}_2,\hat{\bg}_3\}$ coincide, and the image is centered at $\bz_r$,
    \item the cross-range resolution is: $\frac{\lambda|\bz_r|}{a\alpha_r}$ in the $\hat{\bf}_1$ direction and $\frac{\lambda|\bz_r|}{a}$ in the $\hat{\bf}_2$ direction,
    \item \red{when the bandwidth is much smaller, respectively much larger, than
$B_c$, the range resolution is determined by the finite aperture,
respectively by the frequency bandwidth.} More precisely, the range resolution is $\frac{\varepsilon c_0}{2B_H}$ in the broadband case ($B\gg B_c$) and $\frac{\lambda|\bz_r|^2}{a^2}$ in the narrowband case ($B\ll B_c$). \blue{This means that range resolution is determined by the time-harmonic longitudinal Rayleigh formula (see \cite[p.~415]{BornWolf1993}) in the narrowband case and by the bandwidth in the broadband case, when $c_0/B$ becomes smaller than the Rayleigh formula.}
\end{itemize}

The effect caused by the mismatch of $c_s$ and $c_0$ can be summarized in the following.
\begin{itemize}
    \item When the point-like reflector has moderate perpendicular distance to the $\hat{\be}_3$ axis, i.e, the condition \ref{eq: cond_v} is satisfied, the image is not negligible.
    
    \item In the broadband case, the image focuses at the point $\bz_f=\phi_v(\bz_r)$. The vectors $\hat{\bg}_1$ and $\hat{\bg}_2$ define the transverse directions of the image, and $\hat{\bg}_3$ defines the range direction. Furthermore, the resolution is the same as described in the case $c_s=c_0$.
    
    \item In the narrowband case, the mismatch couples with the range direction behavior, which causes further shift of the exact focusing point of the image in the range direction.
    We illustrate this phenomenon by plotting the function $|\cG_{\alpha_r,\alpha_f}|^2$ for three different choices of $c_s$ in Figure \ref{fig: illu_G_34}. 

    \item The mismatch causes the reduction of the peak value at the focusing point for all cases except for the narrowband case with $\alpha_r=1$. This fact paves the way for the estimation of $c_0$ in the general case as discussed in the next section.
    \item The narrowband case with $\alpha_r=1$ is a special case where the range component couples with the mismatch in the following way:
    \begin{align*}
        \blue{\widetilde{\caP}_E(\br_s,c_s)} \sim \left| \cG_{\alpha_r,\alpha_f}\bra{\frac{-a_0\omega_0}{c_0|\bz_r|}\eta_1,
        \frac{-a_0\omega_0}{c_0|\bz_r|}\eta_2,\frac{-a_0^2\omega_0}{c_0|\bz_r|^2}\bra{\eta_3\red{-}|\bz_r|\bra{\bra{\frac{c_0}{c_s}}^2-1}},0}\right|^2.
    \end{align*}
    This means that along the line defined by 
    \(\eta_1 = \eta_2 = 0\) and 
    \(\eta_3 \red{-} |\bz_r| \bigl( (c_0/c_s)^2 - 1 \bigr) = 0\), 
    the point spread function remains at its maximum. Consequently, 
    in this particular case, the peak of the point spread function envelope 
    cannot be used to estimate \(c_0\). Therefore, this scenario is 
    unfavorable for simultaneously estimating the wave speed and imaging 
    the unknown reflector. However, if either the background speed or 
    the reflector position is known, it is still possible 
    to estimate the remaining unknown parameter.

\end{itemize}

\begin{figure}
    \centering
    \includegraphics[width=0.95\linewidth]{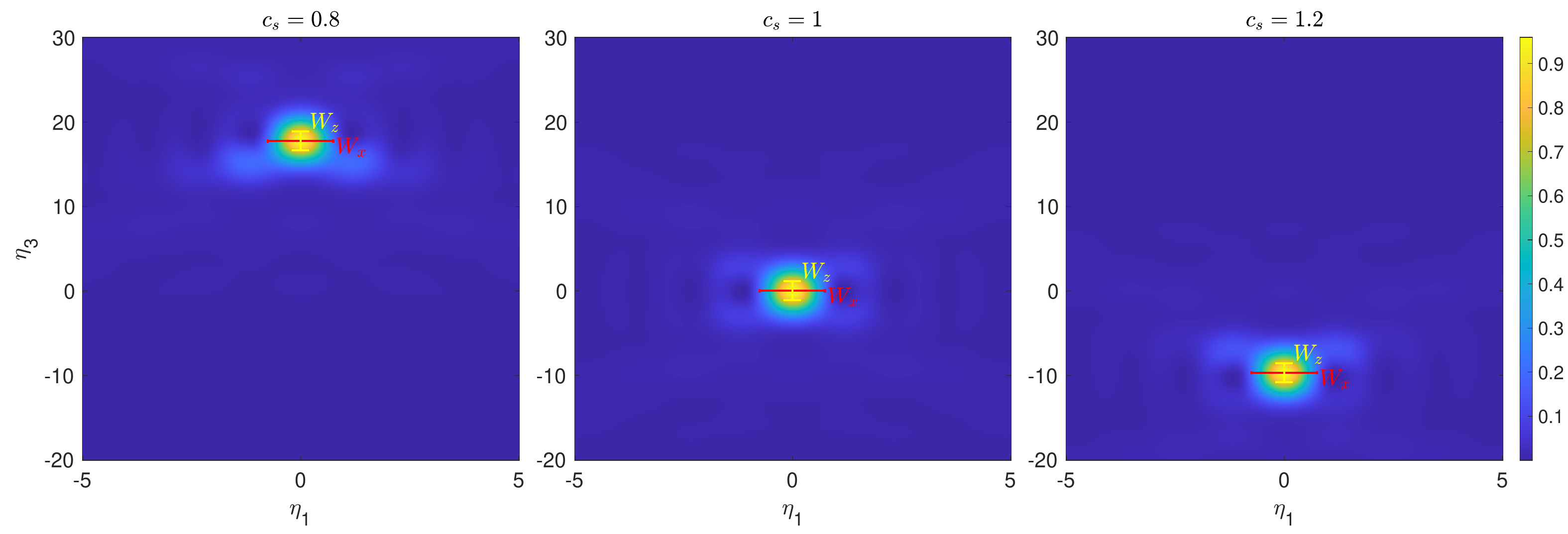}
    \caption{\red{Profiles of the function
$
(\eta_1,\eta_3)\mapsto
\left|
\cG_{\alpha_r,\alpha_f}\left(
\frac{-a_0\omega_0}{c_0|\bz_r|}\eta_1,
0,
\frac{-a_0^2\omega_0}{c_0|\bz_r|^2}\eta_3,
\frac{-a_0^2\omega_0}{c_0|\bz_r|}
\left(
\left(\frac{c_0}{c_s}\right)^2-1
\right)
\right)
\right|^2
$
for three different choices of $c_s$: $c_s=0.8$ (left), $c_s=1$
(middle), and $c_s=1.2$ (right). The peaks of the three profiles are
$0.87$, $1$, and $0.96$, respectively. In all three plots, the marked
widths satisfy
$W_x=4c_0|\bz_r|/(a_0\omega_0)$ and
$W_z=4c_0|\bz_r|^2/(a_0^2\omega_0)$, consistently with the predicted
resolution scaling. Here, the normalized array radius is $a_0=20$,
the normalized central frequency is $\omega_0=4$, the background wave
speed is $c_0=1$, and the reflector position is
$\bz_r=(10,0,20\sqrt{2})$.}}
    \label{fig: illu_G_34}
\end{figure}

\subsection{Estimator for the wave speed}
Previously, we have derived the asymptotic expression of the generalized daylight migration imaging functional when there is a mismatch between the chosen wave speed for backpropagation and the true wave speed in the medium. Based on that, we propose the estimators for the position of the reflector and the wave speed $c_0$ except for the narrowband case with $\alpha_r=1$. We start with the following corollary.

\begin{cor}
    \red{Assume that the function $F_0(t)$ given in \eqref{eq: profile F_1} and \eqref{eq: profile F_2} satisfying the condition that $|F_0(t)|$ achieves its maximum at $\red{t=0}$.} Then the function $\caP_E$ defined in (\ref{eq: envelope wide}), (\ref{eq: envelope wide narrow}), and(\ref{eq: envelope narrow})  achieves its global maximum at $\br_s =0$, $c_s=c_0$.
\end{cor}
\begin{proof}
    We only need to notice that the function $|\cG_{\alpha_r,\alpha_f}(\xi_1,\xi_2,\xi_3,\xi_4)|$ achieves its global maximum at $\xi_j=0$, $j=1,2,3,4$, and the function \red{$|F_{0}(t)|$} achieves its maximum at $\red{t=0}$.
\end{proof}

Based on the corollary above, we propose the following estimator:
\begin{est}
    \begin{align} \label{eq: two estimators}
    (\hat{\bz}_r,\hat{c}) = \operatorname{argmax}_{\bz_s,c_s} \operatorname{Envelope}\{ \cI(\bz_s,c_s)\}.
\end{align}
\end{est}

To implement the proposed estimator, we first note that the generalized daylight migration imaging functional does not directly provide its envelope. To extract the envelope, we apply the Hilbert transform to $\cI(\bz_s,c_s)$. Specifically, we discretize the search domain in spherical coordinates. For each fixed pair of angular variables, the Hilbert transform is performed along the radial direction. This process is then repeated to encompass the entire search domain. The corresponding pseudo-code is provided below.

\begin{algorithm}[!ht]
\caption{Reflector location and wave speed estimation}
\KwIn{Search domain $\Omega \subset \mathbb{R}^3$, \red{a grid of search wave speeds $c_s$}, measurement data $u(t,\bx_r)$.}
\KwOut{The estimators $(\hat{\bz}_r,\hat{c})$}
Calculate the cross-correlation\;
Discretize the search domain $\Omega$ in spherical coordinates $\{(r_j,\theta_k,\varphi_l)\}$\;
\ForEach{fixed searching wave speed $c_s$}{
\ForEach{fixed angular pair $(\theta_k,\varphi_l)$}{
    Perform the Hilbert transform of $\cI((r_j,\theta_k,\varphi_l),c_s)$ along the radial variable $r$\;
    Store the resulting envelope along the radial direction\;
}
}
Assemble the envelopes obtained for all $(\theta_k,\varphi_l)$ to cover the full search domain\;
Calculate the estimators using (\ref{eq: two estimators})\;
\Return{The estimators $(\hat{\bz}_r,\hat{c})$}\;
\end{algorithm}

The resolution of the estimator for the reflector is discussed in Section \ref{sec: imaging functional}. Similarly, based on the theoretical results, we observe that the relative resolution for the wave speed estimation is given by $\frac{\lambda|\bz_r|}{a^2}$ for both narrowband and broadband cases. We discuss this result in the following.
\begin{itemize}
\item Compared with the relative cross-range resolution of reflector localization, given by $\frac{\lambda |\bz_r|}{a}$, the relative resolution for wave speed estimation differs by an additional factor of $|\bz_r|/a$ and the cosine of the tilt angle.

\item In the narrowband case, the relative resolution for wave speed estimation coincides with the relative range resolution for reflector localization. In contrast, in the broadband case, the relative resolution for wave speed estimation is poorer than the relative range resolution for reflector localization.

\item The relative resolution scales inversely with the Fresnel number, defined as
\[
F_{\mathrm{num}} \sim \frac{a^2}{\lambda |\bz_r|}.
\]
When \(F_{\mathrm{num}} \gg 1\), the resolution is high, consistent with the geometric optics approximation adopted in the analysis. When \(F_{\mathrm{num}} \lesssim 1\), diffraction effects become significant, the geometric optics approximation breaks down, and the resolution deteriorates accordingly.
\end{itemize}

\section{Effective wave speed estimation through virtual guide star}\label{sec: random wave speed}
In this section, we develop a method for estimating the effective wave speed in random media, extending the guide-star strategy introduced earlier. 
Since a random medium does not contain a physical guide star, we instead introduce a \emph{virtual guide star} at a location of our own choice within the medium.  
This virtual guide star is not a physical reflector; rather, it is realized indirectly through a family of \red{search points} in the imaging functional.  
The theoretical principles underlying this construction are presented in Section~\ref{subsec: random method}, and the corresponding practical estimator is described in Section~\ref{subsec: random est}.

\subsection{Probing the speckle}\label{subsec: random method}

In practice, our goal is to probe a fixed physical location in the medium, denoted by $\bz_g$ and referred to as the \emph{virtual guide star}. The key observation is that a mismatch between the true propagation speed and the backpropagation speed $c_s$ causes a reflector located at $\bz_r \in \cD$ to focus at a shifted location
\[
\bz_f = \phi_v(\bz_r)
\]
in the imaging domain. \blue{For any fixed $v$ satisfying the condition (\ref{eq: cond_v}), the mapping $\phi_v$, defined in~\eqref{eq: phi_mapping}, is a bijection from the half space $\{\bz:z^{\parallel} \ge 0\}$ onto its image.}

To probe the medium appropriately, we proceed as follows. For a given backpropagation speed $c_s$, we prescribe a travel time $t_g$ and define the corresponding \red{search point} in the imaging functional as
\[
\bz_s(c_s,t_g) = (0,0,c_s\cdot t_g).
\]
With this choice, the associated physical location of the virtual guide star is given by \blue{
\begin{align}\label{eq: def_z_g}
    \bz_g = \phi_v^{-1}\!\bigl(\bz_s(c_s,t_g)\bigr) = (0,0,c_0\cdot t_g),
\end{align}
}
which is independent of $c_s$. Therefore, although the \red{search point} $\bz_s$ varies with the assumed backpropagation speed, it always corresponds to the same physical location $\bz_g$ in the medium.

In principle, the virtual guide star can be placed anywhere within the domain $\cD$ by a similar procedure described above. For simplicity, we restrict our attention to locations along the $\hat{\be}_3$-axis. This is perfectly legitimate as random scatterers are distributed throughout the medium. The corresponding geometric configuration is illustrated in Figure~\ref{fig: random_setup}.

\begin{figure}[!ht]
    \centering
    \begin{tikzpicture} 

    \draw[thick] (3,2.2) rectangle (3.3,3.3); 
    \node at (3.5,2) {$\cA$};               
 
    \foreach \i in {1,...,5} {
        \pgfmathsetmacro\y{2.2 + \i * (3.3 - 2.2)/6}
        \draw (3, \y) -- (3.3, \y);
    }

    \fill (3.15,2.75) circle (2pt);
    \draw[->,thick] (3.15,2.75) -- (3.15,3.75) node[above] {$\hat{\be}_1$};
    \draw[->,thick] (3.15,2.75) -- (4.15,2.75);
    \node[black] at (4.3,3) {$\hat{\be}_3$};
    \node at (2.8,2.6) {$\mathbf{0}$};

    \fill[red] (6,2.75) circle (2pt);
    \node[red] at (6.2,2.4) {$\bz_s$};
    \draw[dashed,thick,red] (3.15,2.75) -- (6,2.75);

    \fill[black] (6.7,2.75) circle (2pt);
    \node[black] at (6.9,3) {$\bz_g$};
    \draw[dashed,thick,black] (6.2,2.75) -- (6.7,2.75);

        \draw[thick, dashed, rounded corners=5pt] (4.0,1) rectangle (8.6,4.6);
        \node at (8.2,4.2) {$\cD$};
        \draw[thick]
            (7.2,1.9) --  
            (7.35,2.2) --  
            (7.55,1.8) --  
            (7.75,2.5) --  
            (7.9,1.7) --  
            (8.1,2.3) --  
            (8.25,2.0);    
        \draw[<->] (7.55,1.6) -- (7.9,1.6) node[midway, below] {$\ell_c$};

    \end{tikzpicture}
    \caption{Illustration for the effective wave speed estimation. For any fixed $c_s$, we pick the \red{search point} at $\bz_s$, and the virtual guide star is positioned at $\bz_g= \phi_v^{-1}(\bz_s)$ with $v = c_s/c_0$. \blue{ Here, $\bz_s$ depends on $c_s$ but $\bz_g$ does not depend on $c_s$.}}
    \label{fig: random_setup}
\end{figure}
The following theorem characterizes the behavior of the imaging functional at the \red{search points} associated with the virtual guide star and the \red{backpropagation} wave speed. \blue{The choice of $\br(\bet)$ follows the similar idea as the parameterization introduced in Theorem~\ref{thm: functional}.}
\begin{thm}\label{thm: random_psf}
    Let us assume that $\hat{F}(\omega)$ has the form defined in (\ref{eq: profile F_1}). For any fixed $c_s$ and a prescribed travel time $t_g$, choosing the \red{\red{search point}}
\[
\bz_s(c_s,t_g) = (0,0,c_s\cdot t_g)
\]
induces a virtual guide star located at \blue{ $\bz_g$ which is defined in (\ref{eq: def_z_g}) and is independent of $c_s$}.
Denote the vector $\bet = (\eta_1,\eta_2,\eta_3)$, and $\br(\bet) = (\varepsilon^{\frac{1}{2}}\eta_1,\varepsilon^{\frac{1}{2}}\eta_2,\varepsilon\eta_3)$, then the imaging functional has the asymptotic form (as $\varepsilon\rightarrow 0$): \blue{
\begin{align}\label{eq: random functional result}
\cI(\bz_s(c_s,t_g),c_s) \approx & \frac{N^2\varepsilon}{2^6\pi^3c_0}
\frac{\cK(\mathbf{0},\bz_g)}{|\bz_g|^2} \cdot \int_{\RR^3} d\bet~  \rho_\mu(\bz_g(t_g)+\br(\bet))\red{\caP_{\rm v}}(t_g,\bet,c_s),
\end{align}}
where $\red{\caP_{\rm v}}(t_g,\bet,c_s)$ is the point spread function defined by 
\begin{align}\label{eq: random psf non_asym}
   \red{\caP_{\rm v}}(t_g,\bet,c_s) =  &\operatorname{Re}\left\{\int_\cB d\omega ~ \bra{-i\omega\hat{F}(\omega)}\exp\bra{i\frac{\omega}{c_0}\bra{2\eta_3+\frac{(\eta_1^2+\eta_2^2)}{|\bz_g|}}}\right. \notag\\
    &\left. \times \cG^2_{1,1}\bra{\frac{a_0\omega}{c_0|\bz_g|}\eta_1,
    \frac{a_0\omega}{c_0|\bz_g|}\eta_2,0,
    \frac{-a_0^2\omega}{c_0|\bz_g|}\bra{\bra{\frac{c_0}{c_s}}^2-1}}\right\}.
\end{align}
If additionally $B_H \ll \omega_0$, then the point spread function has the form 
\begin{align}
    \red{\caP_{\rm v}}(t_g,\bet,c_s) \simeq&  4\pi\omega_0 \operatorname{Re} \left\{
 -i\exp\bra{i\frac{\omega_0}{c_0}\bra{2\eta_3+\frac{(\eta_1^2+\eta_2^2)}{|\bz_g|}}}F_{0}\bra{\frac{B_H}{c_0}\bra{2\eta_3+\frac{(\eta_1^2+\eta_2^2)}{|\bz_g|}}}\right.\notag \\
&\times \left.\cG^2_{1,1}\bra{\frac{a_0\omega_0}{c_0|\bz_g|}\eta_1,
\frac{a_0\omega_0}{c_0|\bz_g|}\eta_2,0,
\frac{-a_0^2\omega_0}{c_0|\bz_g|}\bra{\bra{\frac{c_0}{c_s}}^2-1}} \right\}.
\end{align}
\begin{proof}
    See \ref{append_4.1}.
\end{proof}
\end{thm}
\begin{rmk}\label{rmk: random narrow band}
    In this section, we only consider the case where the bandwidth is not very small compared to the \red{carrier frequency}. We expect that the similar methodology can be extended to the narrowband case as well.
\end{rmk}
We observe that the mismatch between $c_s$ and $c_0$ plays a role in the profile of the \red{imaging functional}, which is a random quantity that depends on the realization of the random medium $\rho_\mu$. The following proposition offers insight that we can extract the information of $c_0$ from the second-order moment of the imaging functional. 
\begin{prop}(second order moment)\label{prop: second order moment}
Following the assumptions in Theorem \ref{thm: random_psf}, we have 
    \begin{align}\label{eq: random_not_narrowband}
        \EE \left[\left| \cI(\bz_s(c_s,t_g),c_s) \right|^2\right] \sim \red{\frac{\ell_c^3}{\varepsilon^2}}\cdot \|\Sigma\|_{L^1(\RR^3)}\cdot\int_{\RR^3}d\bet~ |\red{\caP_{\rm v}(t_g,\bet,c_s)}|^2.
    \end{align}
    If $B_H \ll \omega_0$, we have 
    \begin{align}\label{eq: random narrowband}
        \EE \left[\left| \cI(\bz_s(c_s,t_g),c_s) \right|^2\right] \sim &\int_{\RR^3} d\eta_1 d\eta_2 d\eta_3~\left| F_{0}\bra{\frac{B_H}{c_0}\bra{2\eta_3+\frac{\eta_1^2+\eta_2^2}{|\bz_g|}}} \right|^2 \notag\\
        &\times\left| \cG_{1,1}\bra{\frac{a_0\omega_0}{c_0|\bz_g|}\eta_1,
    \frac{a_0\omega_0}{c_0|\bz_g|}\eta_2,0,
    \frac{-a_0^2\omega_0}{c_0|\bz_g|}\bra{\bra{\frac{c_0}{c_s}}^2-1}} \right|^4 .
    \end{align}
\end{prop}
\begin{proof}
    See \ref{append_4.3}.
\end{proof}
Although the second-order moment provides the information of $c_0$, we should be aware that we have access to only one realization of the random medium in practice. The following two propositions show that by changing the \red{search point} with a distance of the order of the wavelength, the imaging functional is stationary, which links the spatial average to the ensemble average.
\begin{prop}(Locally stationary)\label{prop:stationary}
    Let \red{$\Delta\bz = (\Delta\bz^\perp,\Delta z^{\parallel})$}, and \red{$\bp(\Delta\bz,c_s) = \bra{\frac{\Delta\bz^\perp}{v^2},\frac{\Delta z^{\parallel}}{v}}$}, we have \blue{
    \begin{align}
        \cI(\bz_s(c_s,t_g)+\varepsilon\Delta\bz,c_s) \simeq & \frac{N^2\varepsilon}{2^6\pi^3c_0}
        \frac{\cK(\mathbf{0},\bz_g)}{|\bz_g|^2} \cdot \int_{\RR^3}d\bet~ \rho_\mu(\bz_g+\varepsilon\bp(\Delta\bz,c_s)+\br(\bet))\red{\caP_{\rm v}}(t_g,\bet,c_s).
    \end{align}}
    \red{In particular, at leading order, the imaging functional is locally stationary with respect to the rescaled \red{search point} displacement $\Delta\bz$.}
\end{prop}

The proof of the proposition is almost the same calculation as the one in Theorem \ref{thm: random_psf}. Now, we are able to link the spatial and ensemble averages by showing the ergodicity.
\begin{prop}(spatial average to ensemble average)
\label{prop: ergodic}
\red{Let $\cS(l) = [-\frac{l}{2},\frac{l}{2}]^3$. 
Let $l=l(\varepsilon)$ satisfy
\begin{align}
    l(\varepsilon)\to\infty,
    \qquad
    \varepsilon l(\varepsilon)\to0,
    \qquad \text{as } \varepsilon\to0.
\end{align}
Then
\begin{align}
    \frac{1}{|\cS(l(\varepsilon))|}
    \int_{\cS(l(\varepsilon))}
    d\Delta\bz~
    \left|
    \cI\left(
    \bz_s(c_s,t_g)+\varepsilon\Delta\bz,c_s
    \right)
    \right|^2
    \stackrel{\varepsilon\to0}{\longrightarrow}
    \EE\left[
    \left|
    \cI(\bz_s(c_s,t_g),c_s)
    \right|^2\right]
    \quad \text{in probability.}
\end{align}}
\end{prop}
\begin{proof}
    See \ref{append_4.5}.
\end{proof}

\begin{rmk} \label{rmk:choice_l}
The random field
\[
\Delta\bz\longmapsto
\left|\cI\bigl(\bz_s(c_s,t_g)+\varepsilon\Delta\bz,c_s\bigr)\right|^2
\]
inherits its correlation length from the medium fluctuations.
Since the correlation length of $\rho_\mu$ is $\ell_c$, while the spatial shift
in the medium is of order $\varepsilon\Delta\bz$, the effective correlation
length of the imaging functional in the variable $\Delta\bz$ is
\[
L_{\rm corr}\sim \frac{\ell_c}{\varepsilon}.
\]
Therefore, averaging over a cube $\mathcal S(l)$ amounts to averaging over
approximately
\[
N_{\rm eff}\sim
\left(\frac{\varepsilon l}{\ell_c}\right)^3
\]
effectively independent samples. By the central limit theorem heuristic, the root mean square fluctuation of the
spatial average is expected to scale as
\[
\left(\frac{\ell_c}{\varepsilon l}\right)^{3/2}.
\]Hence, the averaging window should satisfy
\[
l\gg\frac{\ell_c}{\varepsilon},
\]
so that the spatial average provides a reliable approximation of the ensemble
average.

\red{
Moreover, since $\ell_c\ll\varepsilon\ll1$, this requirement is compatible with the condition $\varepsilon l(\varepsilon)\to0$. Indeed, one can choose $l(\varepsilon)$ in the intermediate regime
\[
\frac{\ell_c}{\varepsilon}\ll l(\varepsilon)\ll\frac{1}{\varepsilon}.
\]}
\end{rmk}
In summary, by computing the spatial average of
$|\mathcal{I}(\bz_s(c_s,t_g)+\varepsilon \Delta \bz,c_s)|^2$
over $\Delta \bz \in \mathcal{S}(l)$ with sufficiently large $l$, we obtain an approximation of the ensemble average $\mathbb{E}\left[|\mathcal{I}(\bz_s(c_s,t_g),c_s)|^2\right]$, which contains information about $c_0$. Hence, the spatial average can be used to estimate the effective wave speed.

\subsection{Effective wave speed estimation}\label{subsec: random est}
In this section, we discuss the procedure of the effective wave speed estimation. We start with the following characterization of the second order moment for the regime $B_H\ll \omega_0$.
\begin{cor}\label{cor: random estimator}
    For $\EE\left[|\cI(\bz_s(c_s,t_g),c_s)|^2\right]$ given in (\ref{eq: random narrowband}), we have 
    \[
        c_0 = \arg\max_{c_s} ~  \EE\left[|\cI(\bz_s(c_s,t_g),c_s)|^2\right]
    \]
\end{cor}
\begin{proof}
    See \ref{append_4.6}.
\end{proof}
Then, it is straightforward to propose the following estimator for the effective wave speed based on the spatial average of $\cI(\bz_s(c_s,t_g)+\Delta\bz,c_s)$ as follows:

\begin{est}
For an appropriate choice of $l>0$ \blue{ (see remark \ref{rmk:choice_l}),
    \begin{align}\label{est: effective wave speed}
        \hat{c}_{\textit{eff}} = \operatorname{argmax}_{c_s} ~ \frac{1}{|\cS(l)|}\int_{\cS(l)} \red{d\Delta\bz}~|\cI(\bz_s(c_s,t_g)+\varepsilon\Delta\bz,c_s)|^2 .
    \end{align}}
\end{est}
We summarize the whole procedure of estimating the effective wave speed in the following pseudo-code. 
\begin{algorithm}[!ht]
\caption{Effective wave speed estimation}
\KwIn{\red{A grid of search wave speeds $c_s$}, a grid for spatial average $\{\Delta\bz_j\}_{j=1}^N$, an appropriate choice of $t_g$, typical wave length in the medium $\varepsilon$, measurement data $u(t,\bx_r)$.}
\KwOut{The estimator $\hat{c}_{\textit{eff}}$}
Calculate the cross-correlation\;

\ForEach{fixed searching wave speed $c_s$}{

    Calculate the \red{\red{search point}} $\bz_s(c_s,t_g) = (0,0,c_st_g)$\;
    Calculate the quantity \red{$I = \frac{1}{N}\sum_j|\cI(\bz_s(c_s,t_g)+\varepsilon\Delta\bz_j,c_s)|^2$}\;
}

Calculate the estimators using (\ref{est: effective wave speed})\;
\Return{The estimator $\hat{c}_{\textit{eff}}$}\;
\end{algorithm}

Finally, we notice that the formula (\ref{eq: random narrowband}) indicates that the resolution of the effective wave speed estimator is given by $\frac{\lambda|\bz_g|}{a^2}$.

\section{Numerical experiments}\label{sec: numerical}

In this section, we present the numerical experiments carried out for the estimators introduced in sections \ref{sec: guide star}-\ref{sec: random wave speed}. All simulations are performed in MATLAB on a machine equipped with an M3 8-core CPU. For reference, the computation time for each numerical experiment is reported. We note that all algorithms can be further accelerated using high-performance GPUs, as the code is primarily based on matrix and tensor operations.

\blue{
We assume that the available data is the theoretical cross-correlation \(C^{(1)}\) given in~(\ref{eq: theo cross-cor}), which corresponds to the ensemble-averaged version of the empirical cross-correlation computed from raw recordings. It is known that the empirical cross-correlation converges to its theoretical counterpart as the recording time increases (see, e.g.,~\cite{Garnier_Papanicolaou_2016}), and the two differ only by finite-time statistical fluctuations.
}
\subsection{Homogeneous medium with a point reflector}\label{subsec: numerical homogeneous}
We consider a homogeneous background medium with wave speed $c_0=1$. The distribution of the noise sources is described by
\begin{align}\label{eq: source_region}
K(\bx) = \exp\!\left[-\left(\frac{(x_1+50)^2}{500} + \frac{(x_2+50)^2}{500} + \frac{(x_3-42.5)^2}{10}\right)\right].
\end{align}
The power spectral density is assumed to be $\hat{F}(\omega)=\omega^2 e^{-\omega^2}$. The sensor array lies in the plane \blue{$\operatorname{span}\{\widehat{\be}_1,\widehat{\be}_2\}$} with coordinates $(5j,5k,0)$ for $j,k=-2,-1,\dots,2$. The point reflector is positioned at $(-5,0,50)$ with \red{$\sigma_r \ell_r^3=0.01$}. The numerical configuration is illustrated in Figure~\ref{fig: numerical setup}.

For generating the synthetic wave data, we employ the domain \blue{$[-50,50]\times[-50,50]\times[-50,-35]$} in the $\{\hat{\be}_1,\hat{\be}_2,\hat{\be}_3\}$ coordinates, discretized into a $101\times 101\times 16$ grid for $K(\bx)$. The frequency range is discretized as $0 : 5\times 10^{-3} : 4$, and the time-lag variable $\tau$ is sampled on the grid $40 : 10^{-2} : 300$. The data-generation stage takes approximately 6.4 minutes.

\begin{figure}[!ht]
    \centering
    \begin{tikzpicture} 

  \foreach \i in {0,...,24} {
    \pgfmathsetmacro\y{4.0 - 0.1 * \i}  
    \pgfmathsetmacro\x{rnd * 0.6 + 0.4} 
    \draw (\x, \y) circle (2pt);
  }
 
\draw[thick] (3,2.5) rectangle (3.2,3.1);  
\node at (3.3,2.4) {$\cA$};               

\foreach \i in {1,...,5} {
    \pgfmathsetmacro\y{2.5 + \i * (3.1 - 2.5)/6}
    \draw (3, \y) -- (3.2, \y);
}
  
  \fill (6,2.6) circle (2pt);
  \node at (6.4,2.7) {$\bz_r$};
  
  \draw[-] (0.3,1.5) -- (6.7,1.5);
   \node at (3.4,0.8) {$\be_3$};
  \draw[-] (0.3,1.5) -- (0.3,4.2) ;
   \node at (-0.2,2.85) {$\be_1$};
  \draw (0.3,1.5) -- (0.4,1.5) node[left] {\small $-50$};
  \draw (0.3,2.85) -- (0.4,2.85) node[left] {\small $0$};
  \draw (0.3,4.2) -- (0.4,4.2) node[left] {\small $50$};

  \draw (0.3,1.5) -- (0.3,1.6) node[below=5pt] {\small $-50$};
  \draw (3.15,1.5) -- (3.15,1.6) node[below=5pt] {\small $0$};
  \draw (6,1.5) -- (6,1.6) node[below=5pt] {\small $50$};

    \end{tikzpicture}
    \caption{\red{The configuration of the numerical simulation. The sensor array lies
in the plane $\operatorname{span}\{\widehat{\be}_1,\widehat{\be}_2\}$,
with its center at the origin. The reflector is located at $(-5,0,50)$,
and the circles on the left represent the source region described by
$K(\bx)$, as defined in \eqref{eq: source_region}.}}
    \label{fig: numerical setup}
\end{figure}

\subsubsection{Localization and wave speed estimation}
 We choose the imaging window to be a square of side length $30$ in the plane \blue{$\operatorname{span}\{\widehat{\be}_1,\widehat{\be}_3\}$} centered at the reflector, with a grid step size of $0.5$. Figure~\ref{fig: simu_c_s} displays the numerical results of the generalized daylight migration imaging functional, \red{defined by
\begin{align}\label{eq: simu_imaging_func}
\mathcal{I}_{\mathrm{simu}}(\bz_s,c_s) = \sum_{j,l=1}^N  
    C^{(1)}\!\bra{\cT_s(\bz_s,\bx_l)+\cT_s(\bz_s,\bx_j),\bx_j,\bx_l},
\end{align}
where $\{\bx_j\}_{j=1}^N \subset \mathcal{A}$ denotes the set of sensor locations, $C^{(1)}$ is given in~\eqref{eq: theo cross-cor}, and $\cT_s$ is defined in \eqref{eq: def_time}.} We present the cases $c_s = 0.8$, $c_s = 1$, and $c_s=1.2$. \red{The plots of the imaging functionals are normalized by the maximum value of the imaging functional for $c_s=1$ over the entire search region, and the same normalization factor is used for all three values of $c_s$.} The black diamond indicates the theoretical focusing location given by~(\ref{eq: z_f}). The total computation time for producing the three images is about $3$ seconds. The simulation agrees well with the theoretical prediction, and the shift in the focusing point as $c_s$ varies is clearly observable.

\begin{figure}[!ht]
    \centering
    \includegraphics[width=1\textwidth]{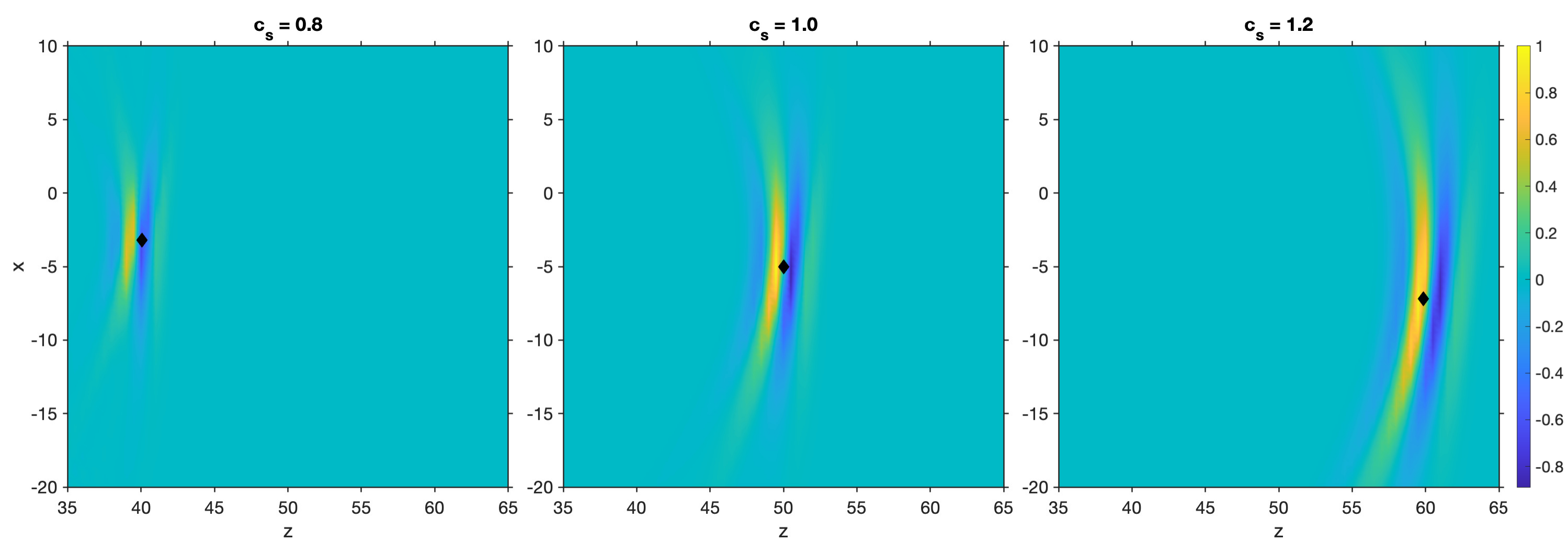} 
    \caption{\blue{Numerical simulations for different values of $c_s$. \red{We plot the imaging functionals defined in \eqref{eq: simu_imaging_func} after normalization.} From left to right, $c_s=0.8$, $c_s=1$, and $c_s=1.2$. The black diamond indicates the theoretically predicted focusing point $\bz_f$. In the middle panel ($c_s=1$), the diamond coincides with the true location of the reflector.}}
    \label{fig: simu_c_s}
\end{figure}

Next, we conduct numerical experiments to estimate the wave speed. 
To this end, we consider the profile of the function
\begin{equation}\label{eq: simu speed est}
\phi_1(c_s) := \max_{\bz_s} \operatorname{Envelope}\{\blue{\cI_{\mathrm{simu}}(\bz_s,c_s)}\},
\end{equation}
obtained from numerical simulations. 

The search domain for $\bz_s$ is discretized in spherical coordinates, with radius ranging from $30$ to $70$, polar angle from $0$ to $12^\circ$, and azimuthal angle from $-8^\circ$ to $8^\circ$, using a $300\times 80\times 80$ grid. The wave speed is sampled from $0.7$ to $1.3$ with step size $0.02$. The entire computation takes approximately $3$ minutes.

For comparison, we also compute the theoretical prediction
\begin{align}\label{eq: theoretical speed estimation}
    \phi_2(c_s) := \operatorname{Envelope}\{\caP(0,c_s)\}
\simeq
\left|\int_{0}^\infty d\omega~
\bra{-i\omega\hat{F}(\omega)}
\cG^2_{\alpha_r,\alpha_f}
\bra{0,0,0,
\frac{-a_0^2\omega}{c_0|\bz_r|}
\left(\left(\frac{c_0}{c_s}\right)^2 - 1\right)
}\right|.
\end{align}

Figure~\ref{fig: wave speed estiamtion} displays both $\phi_1(c_s)$ and $\phi_2(c_s)$, each appropriately normalized. The theoretical and numerical curves agree well, and the wave speed can be accurately recovered by maximizing $\phi_1(c_s)$.
\begin{figure}[!ht]
    \centering
    \includegraphics[width=0.5\textwidth]{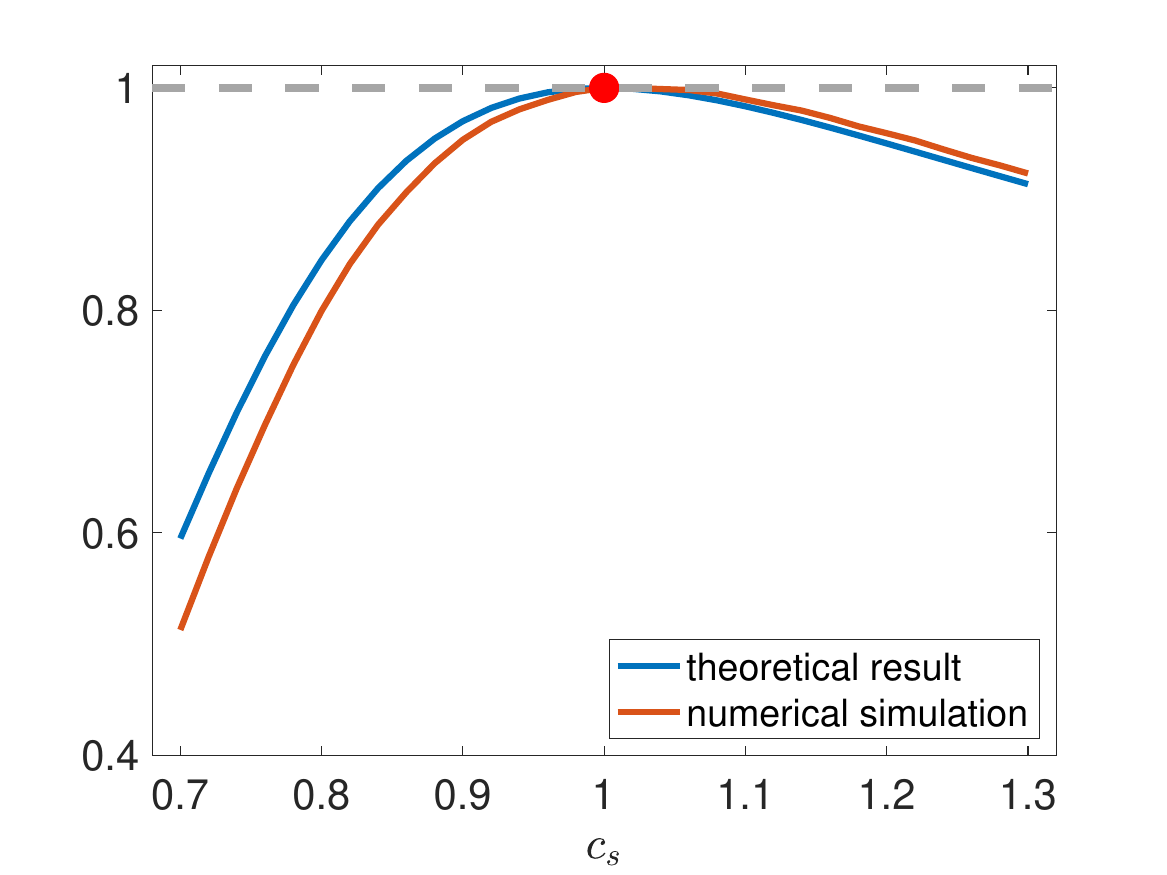}
    \caption{\red{Simulation and theoretical results for the wave speed estimation.
The red line represents $\phi_1(c_s)$, while the blue line represents
$\phi_2(c_s)$. Here, $\phi_1(c_s)$ is defined in
\eqref{eq: simu speed est}, and $\phi_2(c_s)$ is defined in
\eqref{eq: theoretical speed estimation}. For the plot of
$\phi_2(c_s)$, we use $a_0=10$, $\bz_r=(-5,0,50)$, and $c_0=1$,
consistently with the simulation setup. The peak is attained at
$c_s=c_0=1$.}}
    \label{fig: wave speed estiamtion}
\end{figure}

\subsubsection{Monte Carlo experiment}
We conduct Monte Carlo simulations to assess the statistical stability of the proposed wave-speed estimator in a homogeneous background medium. We consider $100$ independent realizations of the noisy sources. For each realization, $500$ sampling points are drawn according to the density function $K(\bx)$, and the corresponding cross-correlations are computed. The discretization grid for the wave data generalization is the same as that used in the previous experiments.

For wave-speed estimation, \red{the grid of search wave speeds $c_s$} is chosen as $0.9{:}0.005{:}1.1$, and the search domain is the same as that used in the previous experiments except for the refinement of the spatial grid in the region of interest to $500 \times 100 \times 100$. The statistical results are described in the following. The estimated mean wave speed is $1.0379$, the sample variance is $1.4151 \times 10^{-4}$, and the root mean square error (RMSE) is $0.0397$. The total computational time for the Monte Carlo simulations is approximately $20$ hours.

We also evaluate the relative localization errors for the reflector using the estimated wave speed and the true wave speed from each Monte Carlo trial. We denote the estimator using estimated wave speed as $\hat{\bz}_{r,est}$, and the one using the true wave speed as $\hat{\bz}_{r,true}$. The relative range and cross-range errors are defined as

\begin{align*}
    &E_{r,est} = \frac{\lb \hat{\bz}_{r,est} , \bz_{r}\rb}{|\bz_r|^2}, \quad E_{c,est} = \frac{\sqrt{|\hat{\bz}_{r,est}-\bz_{r}|^2-E_{r,est}}}{|\bz_r|},\\
    &E_{r,true} = \frac{\lb \hat{\bz}_{r,true} , \bz_{r}\rb}{|\bz_r|^2}, \quad E_{c,true} = \frac{\sqrt{|\hat{\bz}_{r,true}-\bz_{r}|^2-E_{r,true}}}{|\bz_r|}.
\end{align*}
We display the mean, sample variance, and RMSE of these errors in Table~\ref{tab:relative_error_stats}. We observe that the localization error in the range direction is impacted by the uncertainty in the estimated wave speed. The localization error in the cross-range direction, however, is not sensitive to the estimated wave speed. \blue{Although the estimated mean slightly exceeds the theoretical value $(c_0=1)$, the estimator exhibits low Monte Carlo variance, indicating good statistical stability. The remaining discrepancy is primarily due to finite sampling and discretization errors.}

\begin{table}[t]
\centering
\begin{tabular}{cccc}
\toprule
\textbf{Error} & \textbf{Mean} & \textbf{Variance} & \textbf{RMSE} \\
\midrule
$E_{r,est}$

& $3.795\times10^{-2}$ 
& $1.517\times10^{-4}$ 
& $3.988\times10^{-2}$ \\

$E_{r,true}$
& $-9.776\times10^{-4}$ 
& $3.262\times10^{-12}$ 
& $9.776\times10^{-4}$ \\

\midrule
$E_{c,est}$
& $2.366\times10^{-3}$ 
& $1.050\times10^{-6}$ 
& $2.576\times10^{-3}$ \\

$E_{c,true}$
& $2.241\times10^{-3}$ 
& $7.607\times10^{-7}$ 
& $2.403\times10^{-3}$ \\
\bottomrule
\end{tabular}
\caption{Mean, variance, and root mean square error (RMSE) of relative range and cross-range localization errors over 100 Monte Carlo trials. 
All errors are normalized by the true target range.}
\label{tab:relative_error_stats}
\end{table}

\subsection{Random medium with $\ell_c \ll \lambda$}
In this section,we conduct numerical experiments to validate the methodology for estimating the effective wave speed in a random medium, we plot the profile of $\EE \left[\left| \cI(\bz_s(c_s,t_g),c_s) \right|^2\right]$, where the point spread function without the narrowband approximation is given by~(\ref{eq: random narrowband}).  The parameters are set as $a_0=10$, $c_0=1$, and $t_g=50$. We choose
\begin{align}\label{eq: theo_example}
    \red{F(t)} = e^{-\frac{t^2}{2\sigma^2}}\cos(2\pi\nu_0 t),
\qquad \nu_0 = 10, \quad \sigma = 0.3,
\end{align}
for which the corresponding \red{carrier frequency $\omega_0=2\pi\nu_0$ and bandwidth $B_H=\frac{1}{\sigma}$}. The resulting theoretical profile is shown in Figure~\ref{fig: polt_theo_random}, whose peak is attained at $c_s=c_0=1$.
\begin{figure}[!ht]
    \centering
    \includegraphics[width=0.5\linewidth]{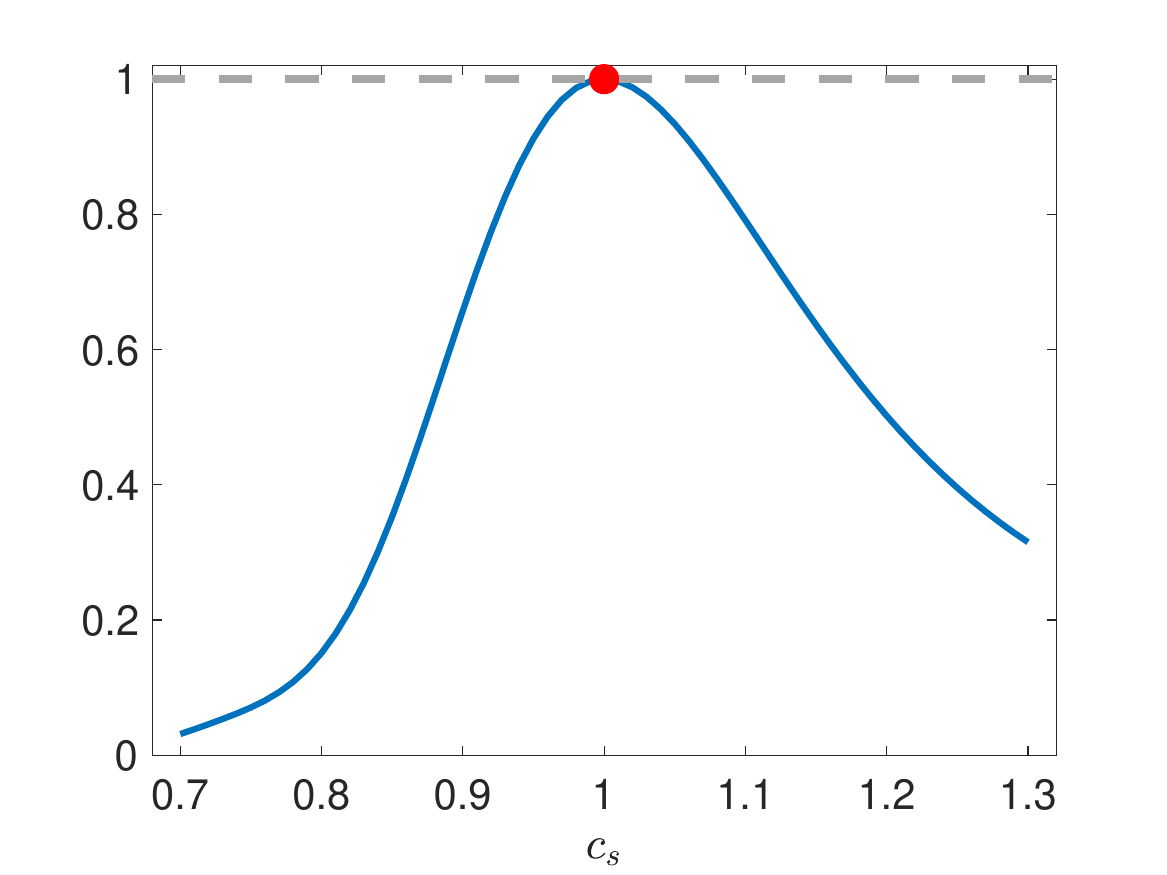}
    \caption{\red{Normalized theoretical profile of
$\EE\left[\left|\cI(\bz_s(c_s,t_g),c_s)\right|^2\right]$
given by \eqref{eq: random_not_narrowband}. We use
$a_0=10$, $c_0=1$, $t_g=50$, and the pulse profile
$F_0(t)$ given by \eqref{eq: theo_example}. The peak is attained at
$c_s=c_0=1$.}}
    \label{fig: polt_theo_random}
\end{figure}

\section{Conclusion}\label{sec: conclusion}

In this paper, we investigated passive imaging in the daylight configuration and developed a mathematical framework to quantitatively analyze the impact of wave-speed mismatch in backpropagation-based reflector imaging. When the searching speed matches the true background speed, our results recover the classical resolution analysis for point-reflector imaging. In the presence of mismatch, we derived explicit formulas characterizing both the spatial shift of the reflector and the defocusing of the associated point spread function. This analysis naturally leads to a novel estimator for the background wave speed.

We further extended the framework to random media and introduced the concept of a virtual guide star, which enables an effective and robust wave-speed estimator under daylight illumination. For clarity of exposition and without loss of generality, our analysis focused on the case where the virtual guide star lies on the axis perpendicular to the sensor array.

\blue{These results demonstrate that the focusing offset and defocusing effects caused by wave speed mismatch in passive imaging should no longer be viewed as a source of imaging error but rather as a useful source of inversion information.}

We expect that the theoretical framework developed here provides a foundation for further advances in wave-speed estimation within the broader field of passive imaging.

\section*{Acknowledgments}

The authors would like to thank the anonymous reviewers for their
helpful comments. This work was partially supported by Agence Nationale
de la Recherche (Grant No. ANR-23-CE30-0021).

\newpage
\appendix
\section{Proofs of the results in Section \ref{sec: guide star}}
\subsection{Proof of Theorem \ref{thm: functional}}\label{append_3.1}
To begin with, we observe that 
\begin{align*}
     C^{(1)}(\tau,\bx_1,\bx_2) =  C^{(1)}(-\tau,\bx_2,\bx_1),
\end{align*}
we write 
\begin{align*}
    \cI(\bz_s,c_s) \approx \cI_{\text{\Rn{1}}+}(\bz_s,c_s)+\cI_{\text{\Rn{2}}+}(\bz_s,c_s)+\cI_{\text{\Rn{1}}-}(\bz_s,c_s)+\cI_{\text{\Rn{2}}-}(\bz_s,c_s),
\end{align*}
where 
\begin{align}\label{eq: I_1_plus_neg}
    \cI_{\text{\Rn{1}}\pm}(\bz_s,c_s) = \frac{N^2}{2|\cA|^2} \iint_{\cA^2} d\bx_1 d\bx_2 \ C^{(1)}_{\text{\Rn{1}}}\bra{\pm(\cT_s(\bz_s,\bx_2)+\cT_s(\bz_s,\bx_1)),\bx_1,\bx_2},
\end{align}
and
\begin{align}\label{eq: I_2_plus_neg}
    \cI_{\text{\Rn{2}}\pm}(\bz_s,c_s) = \frac{N^2}{2|\cA|^2} \iint_{\cA^2} d\bx_1 d\bx_2 \ C^{(1)}_{\text{\Rn{2}}}\bra{\pm(\cT_s(\bz_s,\bx_2)+\cT_s(\bz_s,\bx_1)),\bx_1,\bx_2}.
\end{align}
We first focus on the calculation of $\cI_{\text{\Rn{1}}+}$. We have \blue{
\begin{align*}
     C^{(1)}_{\text{\Rn{1}}}\bra{\tau,\bx_1,\bx_2}=\frac{\sigma_r \ell_r^3}{2^7\pi^4c_0^2\varepsilon^2} \iint d\by d\omega ~\frac{\omega^2\hat{F}(\omega)K(\by)}{|\bx_2-\bz_r||\bz_r-\by||\bx_1-\by|}e^{i\frac{\Phi_{\text{\Rn{1}}}(\omega,\tau,\bx_1,\bx_2,\by)}{\varepsilon}},
\end{align*}}
where \blue{
\begin{align*}
    \Phi_{\text{\Rn{1}}}(\omega,\tau,\bx_1,\bx_2,\by) = \omega(\red{\cT(\bx_2,\bz_r)}+\cT(\bz_r,\by)-\cT(\bx_1,\by)-\tau).
\end{align*}}
Therefore, we derive that
\begin{align}\label{eq: imaging function 1+}
    \cI_{\text{\Rn{1}}+}(\bz_s,c_s) = \frac{\sigma_r \ell_r^3}{2^8\pi^4c_0^2\varepsilon^2}\frac{N^2}{|\cA|^2} \iint_{\cA^2}d\bx_1 d \bx_2\iint d\by d\omega~\frac{\omega^2\hat{F}(\omega)K(\by)}{|\bx_2-\bz_r||\bz_r-\by||\bx_1-\by|}e^{i\frac{\Phi_{\text{\Rn{1}}+}(\omega,\by,\bx_1,\bx_2,\bz_s)}{\varepsilon}},
\end{align}
where 
\begin{align*}
    \Phi_{\text{\Rn{1}}+}(\omega,\by,\bx_1,\bx_2,\bz_s) = \omega(\cT(\bx_2,\bz_r)+\cT(\bz_r,\by)-\cT(\bx_1,\by)-\cT_s(\bx_1,\bz_s)-\cT_s(\bx_2,\bz_s))
\end{align*}
The analysis of the imaging function highly relies on the analysis of the term $\Phi_{\text{\Rn{1}}+}$, in the following, we divide the calculation of the (generalized) travel times into two parts.

First, we consider $\cT(\bx_2,\bz_r)-\cT_s(\bx_2,\bz_s)$.

Let the \red{search point} be
\[
\bz_s = \bz_f + \varepsilon^{\frac{1}{2}}\br,
\]
where $\bz_f$ is the focusing point to be determined. By Taylor expansion, we derive  
\begin{align*}
    &|\bz_r-\bx_2| = |\bz_r| - \varepsilon^{\frac{1}{2}} \frac{\lb \bz_r,\widetilde{\bx_2}\rb}{|\bz_r|}+\cO(\varepsilon),\\
    &\blue{|\bz_s-\bx_2| = |\bz_s| - \varepsilon^{\frac{1}{2}} \frac{\lb \bz_s,\widetilde{\bx_2}\rb}{|\bz_s|}+\cO(\varepsilon) = |\bz_s|  -\varepsilon^{\frac{1}{2}} \frac{\lb \bz_f,\widetilde{\bx_2}\rb}{|\bz_f|}+\cO(\varepsilon).}
\end{align*}
To avoid rapid phase terms in the integral (\ref{eq: imaging function 1+}), we derive
\begin{itemize}
    \item For $\cO(1)$ terms: $\frac{|\bz_r|}{c_0} = \frac{|\bz_f|}{c_s}$.
    \item \blue{For $\cO(\varepsilon^{\frac{1}{2}})$ terms: $\frac{\lb \bz_r,\widetilde{\bx_2}\rb}{c_0|\bz_r|}-\frac{\lb \bz_f,\widetilde{\bx_2}\rb}{ c_s|\bz_f|}=0$ holds for all $\br\in\RR^3$, and $\widetilde{\bx_2}\in \RR^2$.}
\end{itemize}
\blue{The second condition implies that $\lb \bz_f-v^2\bz_r,\hat{\be_2} \rb =0$. By similar arguments, $\lb \bz_f-v^2\bz_r,\hat{\be_1} \rb =0$. The two conditions indicate that $\bz_f^\perp-v^2\bz_r^\perp =0 $.
The calculation above gives the following equations for $\bz_f$:}
\begin{align*}
    |\bz_f| = v|\bz_r|, \quad \bz_f^\perp = v^2\bz_r^\perp,
\end{align*}
whose solution provides the position of $\bz_f$:
\begin{align*}
    \bz_f = \bra{v^2\bz_r^\perp,v\sqrt{(z_r^{\parallel})^2+(1-v^2)|\bz_r^\perp|^2}}, 
\end{align*}
\blue{where $v := \frac{c_s}{c_0}$ satisfies 
\begin{align}
    v  \in \left\{
    \begin{array}{ll}
        \left(0, \sqrt{1+\bra{\frac{z_{r,3}}{|\bz_{r}^\perp|}}^2}\right], & \text{if } \bz_{r}^\perp\ne0,\\
        (0,\infty), & \text{if } \bz_{r}^\perp  =0.
    \end{array}
    \right.
\end{align}}
In the following calculation, we denote $\bz_f=(z_{f,1},0,z_{f,3})$ for simplicity.

Meanwhile, the two conditions above lead us to the following correct scale to proceed the analysis to the profile of imaging function. We consider
\begin{align*}
    \bz_s = \bz_f+\br_s,
\end{align*}
where 
\begin{align*}
    \br_s = \varepsilon^{\frac{1}{2}}v^2\eta_1 \hat{\bg}_1+\varepsilon^{\frac{1}{2}}v^2\eta_2 \hat{\bg}_2+\varepsilon v\eta_3 \hat{\bg}_3.
\end{align*}
We now proceed to the detailed calculation of $\cT(\bx_2,\bz_r)-\cT_s(\bx_2,\bz_s)$. By Taylor's expansion, we get
\begin{align}
   \blue{ |\bz_r-\bx_2| \simeq |\bz_r| - \varepsilon^{\frac{1}{2}} \frac{\lb \bz_r,\widetilde{\bx_2}\rb}{|\bz_r|}+\varepsilon\frac{|\bz_r|^2|\widetilde{\bx_2}|^2-\lb\bz_r,\widetilde{\bx_2}\rb^2}{2|\bz_r|^3},}
\end{align}
and
\begin{align}\label{eq: z_s-x_2}
    |\bz_s-\bx_2| 
    &\blue{\simeq |\bz_s| - \varepsilon^{\frac{1}{2}} \frac{\lb \bz_s,\widetilde{\bx_2}\rb}{|\bz_s|}+\varepsilon\frac{|\bz_s|^2|\widetilde{\bx_2}|^2-\lb\bz_s,\widetilde{\bx_2}\rb^2}{2|\bz_s|^3}}\notag\\
    & = v|\bz_r|-\varepsilon^{\frac{1}{2}}\frac{v\lb\bz_r,\widetilde{\bx_2}\rb}{|\bz_r|}\notag\\
    &~+\frac{\varepsilon}{2v|\bz_r|}\bra{-2v^2(\eta_1\lb\widetilde{\bx_2},\hat{\bg}_1\rb+\eta_2\lb\widetilde{\bx_2},\hat{\bg}_2\rb)+|\widetilde{\bx_2}|^2-\red{v^2\lb\widetilde{\bx_2},\hat{\bf}_3\rb^2}+2v^2|\bz_r|\eta_3+v^4(\eta_1^2+\eta_2^2)}
\end{align}
Hence, we derive that 
\begin{align}
    &\cT(\bx_2,\bz_r)-\cT_s(\bx_2,\bz_s) 
    \simeq -\varepsilon\frac{\eta_3}{c_0}+\varepsilon\frac{|\widetilde{\bx_2}|^2}{2|\bz_r|c_0}\bra{1-\frac{1}{v^2}} +\frac{\varepsilon}{|\bz_r|c_0}\bra{\eta_1\lb\widetilde{\bx_2},\hat{\bg}_1\rb+\eta_2\lb\widetilde{\bx_2},\hat{\bg}_2\rb-\frac{v^2}{2}(\eta_1^2+\eta_2^2)}\notag\\
    &~= -\varepsilon\frac{\eta_3}{c_0}+\varepsilon\frac{|\widetilde{\bx_2}|^2}{2|\bz_r|c_0}\bra{1-\frac{1}{v^2}} +\frac{\varepsilon}{|\bz_r|c_0}\bra{\frac{\widetilde{\bx_{2,1}}z_{f,3}}{|\bz_f|}\eta_1+\widetilde{\bx_{2,2}}\eta_2-\frac{v^2}{2}(\eta_1^2+\eta_2^2)}
\end{align}

Now, we focus on the term $\cT(\bz_r,\by)-\cT(\bx_1,\by)-\cT_s(\bx_1,\bz_s)$.

For the leading order, we have 
\begin{align*}
    \cT(\bz_r,\by)-\cT(\bx_1,\by)-\cT_s(\bx_1,\bz_s) \simeq \frac{|\bz_r-\by|}{c_0}-\frac{|\by|}{c_0}-\frac{|\bz_f|}{c_s} = \frac{|\bz_r-\by|-|\by|-|\bz_r|}{c_0}.
\end{align*}
To avoid the \red{rapid phase} term in the integral (\ref{eq: imaging function 1+}), we parametrize $\by$ as follows:
\begin{align*}
    \by = |\bz_r|\bra{-s_3 \hat{\bf}_3+\varepsilon^{\frac{1}{2}}s_1\hat{\bf}_1+\varepsilon^{\frac{1}{2}}s_2\hat{\bf}_2}.
\end{align*}
The calculation of $|\bz_s-\bx_1|$ is similar to (\ref{eq: z_s-x_2}), and we calculate that
\begin{align}
    |\bz_r-\by| \simeq |\bz_r|\bra{(1+s_3)+\varepsilon\frac{s_1^2+s_2^2}{2(1+s_3)}},
\end{align}
and
\begin{align}
    &|\bx_1-\by| 
    \simeq |\bz_r|s_3+\varepsilon^{\frac{1}{2}}\lb\widetilde{\bx_1},\hat{\bf}_3\rb\notag\\
    &~+\varepsilon\bra{-\frac{\lb\widetilde{\bx_1},\hat{\bf}_3\rb^2}{2|\bz_r|s_3}+\frac{1}{2|\bz_r|s_3}\bra{|\widetilde{\bx_1}|^2+|\bz_r|^2(s_1^2+s_2^2)-2|\bz_r|\left(s_1\lb\widetilde{\bx_1},\hat{\bf}_1\rb)+s_2\lb\widetilde{\bx_1},\hat{\bf}_2\rb\right)  } }
\end{align}
Then, we derive that 
\begin{align}
    &\cT(\bz_r,\by)-\cT(\bx_1,\by)-\cT_s(\bx_1,\bz_s) 
    \simeq -\varepsilon\frac{\eta_3}{c_0}-\varepsilon\frac{v^2(\eta_1^2+\eta_2^2)}{2|\bz_r|c_0}-\varepsilon\frac{|\bz_r|}{c_0}\frac{s_1^2+s_2^2}{2s_3(1+s_3)}+\varepsilon\frac{\lb\widetilde{\bx_1},\bz_r\rb^2}{2|\bz_r|^3c_0}\bra{1+\frac{1}{s_3}} \notag\\
    &~~-\varepsilon\frac{|\widetilde{\bx_1}|^2}{2|\bz_r|c_0}\bra{\frac{1}{s_3}+\frac{1}{v^2}}+\varepsilon\blue{\frac{s_1\lb\widetilde{\bx_1},\hat{\bf}_1\rb+s_2\lb\widetilde{\bx_1},\hat{\bf}_2\rb}{c_0s_3}}+\varepsilon\frac{\eta_1\lb\widetilde{\bx_1},\hat{\bg}_1\rb+\eta_2\lb\widetilde{\bx_1},\hat{\bg}_2\rb}{|\bz_r|c_0}\notag\\
    &~=-\varepsilon\frac{\eta_3}{c_0}-\varepsilon\frac{v^2(\eta_1^2+\eta_2^2)}{2|\bz_r|c_0}-\varepsilon\frac{|\bz_r|}{c_0}\frac{s_1^2+s_2^2}{2s_3(1+s_3)}+\varepsilon\frac{s_1}{c_0s_3}\frac{z_{r,3}}{|\bz_r|}\widetilde{x_{1,1}}+\varepsilon\frac{s_2}{c_0s_3}\widetilde{x_{1,2}}+\varepsilon\frac{\widetilde{x_{1,1}}z_{f,3}}{c_0|\bz_r||\bz_f|}\eta_1+\varepsilon\frac{\widetilde{x_{1,2}}}{|\bz_r|c_0}\eta_2\notag\\
    &~~+\varepsilon\frac{\widetilde{x_{1,1}}^2z^2_{r,1}}{2c_0|\bz_r|^3}\bra{1+\frac{1}{s_3}}-\varepsilon\frac{|\widetilde{\bx_{1}}|^2}{2c_0|\bz_r|}\bra{\frac{1}{s_3}+\frac{1}{v^2}}
\end{align}
We compute the integrals in $s_1$ and $s_2$ \blue{using Lemma \ref{lem: integral}:}
\begin{align*}
    &\iint ds_1 ds_2 ~ \exp{\bra{i\frac{\omega}{\varepsilon}\bra{\cT(\bz_r,\by)-\cT(\bx_1,\by)-\cT_s(\bx_1,\bz_s)}}} \notag \\
    &~\simeq -\frac{2\pi ic_0s_3(1+s_3)}{\omega|\bz_r|}\exp\bra{i\frac{\omega}{c_0}\bra{\frac{|\widetilde{\bx_1}|^2}{2|\bz_r|}\bra{1-\frac{1}{v^2}}-\eta_3-\frac{v^2(\eta_1^2+\eta_2^2)}{2|\bz_r|c_0}+\frac{\widetilde{x_{1,1}}z_{f,3}}{|\bz_f||\bz_r|}\eta_1+\frac{\widetilde{x_{1,2}}}{|\bz_r|}\eta_2}}.
\end{align*}
It gives 
\begin{align*}
     &\iint ds_1 ds_2 ~ \red{e^{i\frac{\Phi_{\text{\Rn{1}}+}(\omega,\by,\bx_1,\bx_2,\bz_s)}{\varepsilon}}} \simeq
     -\frac{2\pi ic_0s_3(1+s_3)}{\omega|\bz_r|}\notag\\
     &~\times\exp\bra{i\frac{\omega}{c_0}\bra{\frac{|\widetilde{\bx_1}|^2+|\widetilde{\bx_2}|^2}{2|\bz_r|}\bra{1-\frac{1}{v^2}}-2\eta_3-\frac{v^2(\eta_1^2+\eta_2^2)}{|\bz_r|}+\frac{z_{f,3}\eta_1}{|\bz_f||\bz_r|}(\widetilde{x_{1,1}}+\widetilde{\bx_{2,1}})+\frac{\eta_2}{|\bz_r|}(\widetilde{x_{1,2}}+\widetilde{\bx_{2,2}})}}.
\end{align*}

Therefore, we obtain
\begin{align*}
    \cI_{\text{\Rn{1}}+}(\bz_s,c_s) \blue{\approx}& \frac{\sigma_r\ell_r^3N^2\cK(\mathbf{0},\bz_r)}{\blue{2^7}\pi^3c_0|\bz_r|^2\varepsilon}\int_\cB d\omega ~ \bra{-i\omega\hat{F}(\omega)}\exp\bra{-i\frac{\omega}{c_0}\bra{2\eta_3+\frac{c_s^2(\eta_1^2+\eta_2^2)}{c_0^2|\bz_r|}}} \notag\\
    & \times \cG^2_{\alpha_r,\alpha_f}\bra{\frac{-a_0\omega}{c_0|\bz_r|}\eta_1,
\frac{-a_0\omega}{c_0|\bz_r|}\eta_2,0,
\frac{-a_0^2\omega}{c_0|\bz_r|}\bra{\bra{\frac{c_0}{c_s}}^2-1}}
\end{align*}
Similar calculation can be done for $\cI_{\text{\Rn{2}}+}$, $\cI_{\text{\Rn{1}}-}$, and $\cI_{\text{\Rn{2}}-}$. Only $\cI_{\text{\Rn{2}}-}$ provides non-vanishing contribution. Hence, the expression in (\ref{eq: functional result}) is derived.

\subsection{Proof of Proposition \ref{prop: functional}}\label{append_3.2}
We revise the parametrization of the \red{search point} as $$\bz_s = \bz_f + \varepsilon^{\frac{1}{2}}v^2\eta_1 \hat{\bg}_1+\varepsilon^{\frac{1}{2}}v^2\eta_2 \hat{\bg}_2+v\eta_3 \hat{\bg}_3.$$
\red{In the following calculation, the relationship $|\bz_f| = v|\bz_r|$ will be frequently used.}
Compared to the previous one, the scaling in direction $\hat{\bg}_3$ is changed. We conduct a similar calculation and get 

\begin{align}\label{eq: z_s-x_2 new}
    |\bz_s-\bx_2| 
    &\simeq |\bz_s| - \varepsilon^{\frac{1}{2}} \frac{\lb \bz_s,\widetilde{\bx_2}\rb}{|\bz_s|}+\varepsilon\frac{|\bz_s|^2|\widetilde{\bx_2}|^2-\lb\bz_s,\widetilde{\bx_2}\rb^2}{\red{2|\bz_s|^3}}\notag\\
    & = |\bz_f|+v\eta_3-\varepsilon^{\frac{1}{2}}\lb\widetilde{\bx_2},\hat{\bg}_3\rb \notag \\
    &~\ +\varepsilon\bra{\frac{v^4(\eta_1^2+\eta_2^2)}{2(|\bz_f|+v\eta_3)} -\frac{%
    v^2\left( \eta_1\langle \widetilde{\bx_2}, \hat{\bg}_1 \rangle 
    + \eta_2\langle \widetilde{\bx_2}, \hat{\bg}_2 \rangle \right)%
    }{|\bz_f| + v\eta_3}+\frac{|\widetilde{\bx_2}|^2-\lb\widetilde{\bx_2},\hat{\bg}_3\rb^2}{2(|\bz_f|+v\eta_3)}}
\end{align}
Then, we have 
\begin{align}
    \cT(\bx_2,\bz_r)-\cT_s(\bx_2,\bz_s) 
    \simeq
    & \blue{-\frac{\eta_3}{c_0}+\frac{\varepsilon}{c_0}\left(\frac{|\widetilde{\bx_2}|^2}{2|\bz_r|}-\frac{|\widetilde{\bx_2}|^2}{2v(|\bz_f|+v\eta_3)}+\frac{\lb\widetilde{\bx_2},\hat{\bg}_3\rb^2}{2v(|\bz_f|+v\eta_3)}-\frac{\lb\widetilde{\bx_2},\hat{\bf}_3\rb^2}{2|\bz_r|}\right.}\notag\\
    &~\left. +\frac{v\left( \eta_1\langle \widetilde{\bx_2}, \hat{\bg}_1 \rangle 
    + \eta_2\langle \widetilde{\bx_2}, \hat{\bg}_2 \rangle \right)}{|\bz_f| + v\eta_3} -\frac{v^3(\eta_1^2+\eta_2^2)}{2(|\bz_f| + v\eta_3)}\right) \notag\\
    &= -\frac{\eta_3}{c_0}+\frac{\varepsilon}{c_0}\left(\frac{|\widetilde{\bx_2}|^2}{2}\bra{\frac{1}{|\bz_r|}-\frac{1}{v^2(|\bz_r|+\eta_3)}}+\frac{\widetilde{\bx_{2,1}}^2z_{r,1}^2}{2\red{|\bz_r|^2}}\bra{\frac{1}{|\bz_r|+\eta_3}-\frac{1}{|\bz_r|}}\right.\notag\\
    &~\left. \red{+}\frac{\widetilde{\bx_{2,1}}}{|\bz_r|+\eta_3}\frac{z_{f,3}}{|\bz_f|}\eta_1+\frac{\widetilde{\bx_{2,2}}}{|\bz_r|+\eta_3}\eta_2-\frac{v^2(\eta_1^2+\eta_2^2)}{2(|\bz_r|+\eta_3)}\right),
\end{align}
and
\begin{align}
    &\cT(\bz_r,\by)-\cT(\bx_1,\by)-\cT_s(\bx_1,\bz_s) 
    \simeq -\frac{\eta_3}{c_0}+\frac{\varepsilon}{c_0}\left( -\frac{(s_1^2+s_2^2)|\bz_r|}{2s_3(1+s_3)}+\frac{\widetilde{x_{1,1}}^2z_{r,1}^2}{2\red{|\bz_r|^2}}\bra{\frac{1}{s_3|\bz_r|}+\frac{1}{|\bz_r|+\eta_3}}\right.\notag\\
    &~\left.-\frac{|\widetilde{\bx_1}|^2}{2}\bra{\frac{1}{|\bz_r|s_3}+\frac{1}{v^2(|\bz_r|+\eta_3)}}+\blue{\frac{s_1}{s_3}\frac{z_{r,3}}{|\bz_{r}|}\widetilde{x_{1,1}}}+\frac{s_2}{s_3}\widetilde{x_{1,2}}-\frac{v^2(\eta_1^2+\eta_2^2)}{2(|\bz_r|+\eta_3)}+\frac{z_{f,3}}{|\bz_f|}\frac{\widetilde{x_{1,1}}\eta_1}{|\bz_r|+\eta_3}+\frac{\widetilde{x_{1,2}}\eta_2}{|\bz_r|+\eta_3}\right).
\end{align}
Computing the integral in $s_1$ and $s_2$ \blue{using Lemma \ref{lem: integral}}, we derive 
\begin{align*}
    &\iint ds_1 ds_2 ~ \red{e^{i\frac{\Phi_{\text{\Rn{1}}+}(\omega,\by,\bx_1,\bx_2,\bz_s)}{\varepsilon}}} \simeq -\frac{2\pi ic_0s_3(1+s_3)}{\omega|\bz_r|}\exp\bra{i\frac{\omega}{c_0}
    \bra{
    -2\frac{\eta_3}{\varepsilon}-\frac{v^2(\eta_1^2+\eta_2^2)}{|\bz_r|+\eta_3}+\varphi(\widetilde{\bx_1})+\varphi(\widetilde{\bx_2})
    }
    },
\end{align*}
where, for $\bx = (x_1,x_2)\in \RR^2 $, $\varphi(\bx)$ is defined by
\begin{align}
    \varphi(\bx) = \frac{z_{f,3}}{|\bz_f|}\frac{x_1\eta_1}{|\bz_r|+\eta_3}+\frac{x_2\eta_2}{|\bz_r|+\eta_3}+\bra{\frac{x_1^2}{2}\frac{z_{r,3}^2}{|\bz_r|^2}+\frac{x_2^2}{2}}\frac{\eta_3}{|\bz_r|(|\bz_r|+\eta_3)}+\frac{|\bx|^2}{2}\frac{1-\frac{1}{v^2}}{|\bz_r|+\eta_3}.
\end{align}
Therefore, we have 
\begin{align*}
    \cI_{\text{\Rn{1}}+}(\bz_s,c_s) \blue{\approx}& \frac{\sigma_r\ell_r^3N^2\cK(\mathbf{0},\bz_r)}{\blue{2^7}\pi^3c_0|\bz_r|^2\varepsilon}\int_\cB d\omega ~ \bra{-i\omega\hat{F}(\omega)}\exp\bra{-i\frac{\omega}{c_0}\bra{2\frac{\eta_3}{\varepsilon}+\frac{c_s^2(\eta_1^2+\eta_2^2)}{c_0^2(|\bz_r|+\eta_3)}}} \notag\\
    & \times \cG^2_{\alpha_r,\alpha_f}\bra{\frac{-a_0\omega}{c_0(|\bz_r|+\eta_3)}\eta_1,
\frac{-a_0\omega}{c_0(|\bz_r|+\eta_3)}\eta_2,\frac{-a_0^2\omega}{c_0|\bz_r|(|\bz_r|+\eta_3)}\eta_3,
\frac{-a_0^2\omega}{c_0(|\bz_r|+\eta_3)}\bra{\bra{\frac{c_0}{c_s}}^2-1}}
\end{align*}
Similar calculation can be done for $\cI_{\text{\Rn{2}}+}$, $\cI_{\text{\Rn{1}}-}$, and $\cI_{\text{\Rn{2}}-}$. Only $\cI_{\text{\Rn{2}}-}$ provides non-vanishing contribution. Hence, we have
\begin{align}
    \cI(\bz_s,c_s) \approx& \frac{\sigma_r\ell_r^3N^2\cK(\mathbf{0},\bz_r)}{2^6\pi^3c_0|\bz_r|^2\varepsilon}
    \operatorname{Re}\left\{\int_\cB d\omega ~ \bra{-i\omega\hat{F}(\omega)}\exp\bra{-i\frac{\omega}{c_0}\bra{2\frac{\eta_3}{\varepsilon}+\frac{c_s^2(\eta_1^2+\eta_2^2)}{c_0^2(|\bz_r|+\eta_3)}}}\right. \notag\\
    &\left. \times \cG^2_{\alpha_r,\alpha_f}\bra{\frac{-a_0\omega}{c_0(|\bz_r|+\eta_3)}\eta_1,
\frac{-a_0\omega}{c_0(|\bz_r|+\eta_3)}\eta_2,\frac{-a_0^2\omega}{c_0|\bz_r|(|\bz_r|+\eta_3)}\eta_3,
\frac{-a_0^2\omega}{c_0(|\bz_r|+\eta_3)}\bra{\bra{\frac{c_0}{c_s}}^2-1}}\right\}
\end{align}

\blue{Further, by the assumption that the function $\hat{F}(\omega)$ has the following form:}
\begin{align*}
    \hat{F}(\omega) = \frac{1}{\varepsilon \red{B_0}}\bra{\hat{F}_{0}\bra{\frac{\omega_0-\omega}{\varepsilon \red{B_0}}}+\hat{F}_{0}\bra{\frac{\omega_0+\omega}{\varepsilon \red{B_0}}} },
\end{align*}
we derive the following expression \blue{
\begin{align}
    \cI(\bz_s,c_s) \approx& \red{\frac{\sigma_r\ell_r^3N^2\cK(\mathbf{0},\bz_r)\omega_0}{2^4\pi^2c_0|\bz_r|^2\varepsilon}}
    \operatorname{Re}\left\{ -i\exp\bra{-i\frac{\omega_0}{c_0}\bra{2\frac{\eta_3}{\varepsilon}+\frac{c_s^2(\eta_1^2+\eta_2^2)}{c_0^2(|\bz_r|+\eta_3)}}} F_{0}\bra{-\frac{2B_0\eta_3}{c_0}}\right. \notag\\
    &\left. \times \cG^2_{\alpha_r,\alpha_f}\bra{\frac{-a_0\omega_0}{c_0(|\bz_r|+\eta_3)}\eta_1,
\frac{-a_0\omega_0}{c_0(|\bz_r|+\eta_3)}\eta_2,\frac{-a_0^2\omega_0}{c_0|\bz_r|(|\bz_r|+\eta_3)}\eta_3,
\frac{-a_0^2\omega_0}{c_0(|\bz_r|+\eta_3)}\bra{\bra{\frac{c_0}{c_s}}^2-1}}\right\}
\end{align}}
The condition that $|\br_s|\ll|\bz_r|$ gives the expression (\ref{eq: functional narrow}). When $B_0\ll B_c$, we get the expression (\ref{eq: functional narrow smaller}). When $B_0\gg B_c$, we get the expression (\ref{eq: functional narrow larger}).

\section{Proofs of the results in Section \ref{sec: random wave speed}}
\subsection{Proof of Theorem \ref{thm: random_psf}}\label{append_4.1}
\blue{Similar to the proof of Theorem \ref{thm: functional}, we have
\begin{align*}
    \cI(\bz_s,c_s) \approx \cI_{\text{\Rn{1}}+}(\bz_s,c_s)+\cI_{\text{\Rn{2}}+}(\bz_s,c_s)+\cI_{\text{\Rn{1}}-}(\bz_s,c_s)+\cI_{\text{\Rn{2}}-}(\bz_s,c_s),
\end{align*}
where $\cI_{\text{\Rn{1}}\pm}(\bz_s,c_s)$ and $\cI_{\text{\Rn{2}}\pm}(\bz_s,c_s)$ are given by (\ref{eq: I_1_plus_neg}) and (\ref{eq: I_2_plus_neg}) respectively.}

We first focus on the calculation of $\cI_{\text{\Rn{1}}+}$. We have
\begin{align*}
    C^{(1)}_{\text{\Rn{1}}}\bra{\tau,\bx_1,\bx_2}=\frac{1}{2^7\pi^4c_0^2\varepsilon^2} \int_{\cD}d\bz \rho_\mu(\bz) \iint d\by d\omega ~\frac{\omega^2\hat{F}(\omega)K(\by)}{|\bx_2-\bz||\bz-\by||\bx_1-\by|}e^{i\frac{\red{\Phi_{\text{\Rn{1}}}(\omega,\bx_1,\bx_2,\bz,\by,\tau)}}{\varepsilon}},
\end{align*}
where 
\begin{align*}
    \red{\Phi_{\text{\Rn{1}}}(\omega,\bx_1,\bx_2,\bz,\by,\tau) = \omega(\cT(\bx_2,\bz)+\cT(\bz,\by)-\cT(\bx_1,\by)-\tau).}
\end{align*}
Therefore, we derive that
\begin{align}
    \blue{\cI_{\text{\Rn{1}}+}(\bz_s,c_s) = \frac{1}{2^8\pi^4c_0^2\varepsilon^2}\frac{N^2}{|\cA|^2} \iint_{\cA^2}d\bx_1d\bx_2\int_{\cD}\rho_\mu(\bz)d\bz\iint d\by d\omega~\frac{\omega^2\hat{F}(\omega)K(\by)e^{i\frac{\Phi_{\text{\Rn{1}}+}(\omega,\by,\bx_1,\bx_2,\bz_s,\bz)}{\varepsilon}}}{\red{|\bx_2-\bz||\bz-\by||\bx_1-\by|}}},
\end{align}
where 
\begin{align*}
    \Phi_{\text{\Rn{1}}+}(\omega,\by,\bx_1,\bx_2,\bz_s,\bz) = \omega(\cT(\bx_2,\bz)+\cT(\bz,\by)-\cT(\bx_1,\by)-\cT_s(\bx_1,\bz_s)-\cT_s(\bx_2,\bz_s)).
\end{align*}
The analysis of the imaging function highly relies on the analysis of the term $\Phi_{\text{\Rn{1}}+}$, in the following, we divide the calculation of the (generalized) travel times into two parts.

First, we consider $\cT(\bx_2,\bz)-\cT_s(\bx_2,\bz_s)$:

By Taylor expansion, we have 
\begin{align}
    |\bx_2-\bz_s |\simeq |\bz_s| + \frac{|\widetilde{\bx_2}|^2}{2|\bz_s|}\varepsilon,
\end{align}
and 
\begin{align}
    |\bx_2-\bz| \simeq |\bz_g|+\frac{\varepsilon}{2|\bz_g|}\bra{\eta_1^2+\eta_2^2+2\eta_3|\bz_g|+|\widetilde{\bx_2}|^2-2(\eta_1\lb\hat{\be}_1,\widetilde{\bx_2}\rb+\eta_2\lb\red{\hat{\be}_2},\widetilde{\bx_2}\rb)}
\end{align}
Then, we have 
\begin{align}
    \cT(\bx_2,\bz)-\cT_s(\bx_2,\bz_s) \simeq \frac{\varepsilon}{c_0}\bra{\frac{\eta_1^2+\eta_2^2}{2|\bz_g|}+\eta_3+\frac{|\widetilde{\bx_2}|^2}{2|\bz_g|}\bra{1-\frac{1}{v^2}}-\frac{1}{|\bz_g|}(\eta_1\lb\hat{\be}_1,\widetilde{\bx_2}\rb+\eta_2\lb\hat{\be_2},\widetilde{\bx_2}\rb))}
\end{align}
To calculate the term $\cT(\bz,\by)-\cT(\bx_1,\by)-\cT_s(\bx_1,\bz_s)$, we introduce the parametrization of $\by$ as:
\begin{align*}
    \by = |\bz_g|\bra{-s_3 \hat{\be}_3+\varepsilon^{\frac{1}{2}}s_1\hat{\be}_1+\varepsilon^{\frac{1}{2}}s_2\hat{\be}_2}.
\end{align*}
By Taylor expansion, we have
\begin{align}
    &|\bx_1-\bz_s| \simeq |\bz_s|+\frac{|\widetilde{\bx_1}|^2}{2|\bz_s|}\varepsilon,\\
    &\blue{|\bz-\by| \simeq (1+s_3)|\bz_g|+\frac{\varepsilon}{2(1+s_3)|\bz_g|}\bra{2\eta_3(1+s_3)|\bz_g|+(\eta_1-\red{|\bz_g|}s_1)^2+(\eta_2-\red{|\bz_g|}s_2)^2}},\\
    &|\bx_1-\by| \simeq |\bz_g|s_3+\varepsilon\bra{\frac{|\widetilde{\bx_1}|^2}{2|\bz_g|s_3}+\frac{(s_1^2+s_2^2)|\bz_g|}{2s_3}-\frac{s_1\lb\widetilde{\bx_1},\hat{\be}_1\rb+s_2\lb\widetilde{\bx_1},\hat{\be}_2\rb}{s_3}}.
\end{align}
Hence, we have 
\begin{align}
    \cT(\bz,\by)-\cT(\bx_1,\by)&-\cT_s(\bx_1,\bz_s) \simeq\frac{\varepsilon}{c_0}\left(-\frac{|\widetilde{\bx_1}|^2}{2|\bz_g|}\bra{\frac{1}{v^2}+\frac{1}{s_3}}-\frac{(s_1^2+s_2^2)|\bz_g|}{2s_3}+\frac{s_1\lb\widetilde{\bx_1},\hat{\be}_1\rb+s_2\lb\widetilde{\bx_1},\hat{\be}_2\rb}{s_3}\right. \notag \\
    &\left.+\eta_3+\frac{1}{2(1+s_3)|\bz_g|}\bra{|\bz_g|^2(s_1^2+s_2^2)-2|\bz_g|(s_1\eta_1+s_2\eta_2)+(\eta_1^2+\eta_2^2)} \right).
\end{align}
Computing the integral in $s_1$ and $s_2$ \blue{using Lemma \ref{lem: integral}}, we derive 
\begin{align*}
    \iint ds_1 ds_2 ~ \red{e^{i\frac{\Phi_{\text{\Rn{1}}+}(\omega,\by,\bx_1,\bx_2,\bz_s,\bz)}{\varepsilon}} }\simeq -\frac{2\pi ic_0s_3(1+s_3)}{\omega|\bz_g|}\exp\bra{i\frac{\omega}{c_0}
    \bra{
    2\eta_3+\frac{\eta_1^2+\eta_2^2}{|\bz_g|}+\psi(\widetilde{\bx_1})+\psi(\widetilde{\bx_2})
    }
    },
\end{align*}
where, for $\bx = (x_1,x_2)\in \RR^2 $, \red{$\psi(\bx)$} is defined by
\begin{align}
    \psi(\bx) = \frac{|\bx|^2}{2|\bz_g|}\bra{1-\frac{1}{v^2}}-\frac{\eta_1x_1+\eta_2x_2}{|\bz_g|}.
\end{align}

Therefore, we obtain \blue{
\begin{align*}
    \cI_{\text{\Rn{1}}+}(\bz_s,c_s) =& \frac{N^2\cK(\mathbf{0},\bz_g)\varepsilon}{2^7\pi^3c_0|\bz_g|^2}\int d \bet~ \rho_\mu(\bz_g(t_g)+\br(\bet)) \int_\cB d\omega ~ \bra{-i\omega\hat{F}(\omega)}\exp\bra{i\frac{\omega}{c_0}\bra{2\eta_3+\frac{(\eta_1^2+\eta_2^2)}{|\bz_g|}}} \notag\\
    & \times \cG^2_{1,1}\bra{\frac{a_0\omega}{c_0|\bz_g|}\eta_1,
    \frac{a_0\omega}{c_0|\bz_g|}\eta_2,0,
    \frac{-a_0^2\omega}{c_0|\bz_g|}\bra{\bra{\frac{c_0}{c_s}}^2-1}}.
 \end{align*}}
Similar calculation can be done for $\cI_{\text{\Rn{2}}+}$, $\cI_{\text{\Rn{1}}-}$, and $\cI_{\text{\Rn{2}}-}$. Only $\cI_{\text{\Rn{2}}-}$ provides non-vanishing contribution. Hence, we have \blue{
\begin{align*}
    \cI(\bz_s,c_s) =& \frac{N^2\cK(\mathbf{0},\bz_g)\varepsilon}{2^6\pi^3c_0|\bz_g|^2}\int d \bet~  \rho_\mu(\bz_g(t_g)+\br(\bet)) \operatorname{Re}\left\{\int_\cB d\omega ~ \bra{-i\omega\hat{F}(\omega)}\exp\bra{i\frac{\omega}{c_0}\bra{2\eta_3+\frac{(\eta_1^2+\eta_2^2)}{|\bz_g|}}}\right. \notag\\
    &\left. \times \cG^2_{1,1}\bra{\frac{a_0\omega}{c_0|\bz_g|}\eta_1,
    \frac{a_0\omega}{c_0|\bz_g|}\eta_2,0,
    \frac{-a_0^2\omega}{c_0|\bz_g|}\bra{\bra{\frac{c_0}{c_s}}^2-1}}\right\} .
\end{align*}
}
\subsection{Proof of Proposition \ref{prop: second order moment}}\label{append_4.3}

We observe that
\begin{align*}
    \EE \left[\left| \cI(\bz_s,c_s) \right|^2\right]
    &\sim
    \int_{\RR^3}\int_{\RR^3}
    d\bet\,d\bet'~
    \EE\left[
    \rho_\mu(\bz_g+\br(\bet))
    \rho_\mu(\bz_g+\br(\bet'))
    \right]
    \cdot
    \red{\caP_{\rm v}}(t_g,\bet,c_s)
    \overline{\red{\caP_{\rm v}}(t_g,\bet',c_s)}
    \\
    &=
    \int_{\RR^3}\int_{\RR^3}
    d\bet\,d\bet'~
    \Sigma\left(
    \frac{\br(\bet)-\br(\bet')}{\ell_c}
    \right)
    \cdot
    \red{\caP_{\rm v}}(t_g,\bet,c_s)
    \overline{\red{\caP_{\rm v}}(t_g,\bet',c_s)}
    \\
    &\simeq
    \frac{\ell_c^3}{\varepsilon^{2}}\cdot
    \|\Sigma\|_{L^1(\RR^3)}\cdot
    \int_{\RR^3}d\bet~
    |\red{\caP_{\rm v}}(t_g,\bet,c_s)|^2,
\end{align*}
where we have used Lemma~\ref{lem: cov} in the last equality.

To calculate the profile of the point spread function for
$B_H\ll\omega_0$, we first denote
\[
\phi(\eta_1,\eta_2,\eta_3)
=
\frac{1}{c_0}
\left(
2\eta_3+
\frac{\eta_1^2+\eta_2^2}{|\bz_g|}
\right),
\]
and
\[
\cG(\eta_1,\eta_2;c_s)
=
\cG_{1,1}\left(
\frac{a_0\omega_0}{c_0|\bz_g|}\eta_1,
\frac{a_0\omega_0}{c_0|\bz_g|}\eta_2,
0,
-\frac{a_0^2\omega_0}{c_0|\bz_g|}
\left(
\left(\frac{c_0}{c_s}\right)^2-1
\right)
\right).
\]
Then, we define
\[
f(\eta_1,\eta_2,\eta_3)
=
-4\pi\omega_0 i\,
\exp\left\{
i\omega_0\phi(\eta_1,\eta_2,\eta_3)
\right\}
F_0\left(
B_H\phi(\eta_1,\eta_2,\eta_3)
\right)
\cG^2(\eta_1,\eta_2;c_s).
\]

Since
\[
\red{\caP_{\rm v}}(t_g,\bet,c_s)
=
\operatorname{Re}\{f(\bet)\},
\]
we have
\begin{align}
    \int_{\RR^3}d\bet~
    |\red{\caP_{\rm v}}(t_g,\bet,c_s)|^2
    &=
    \frac{1}{2}
    \int_{\RR^3}d\bet~
    |f(\bet)|^2
    +
    \frac{1}{2}
    \operatorname{Re}
    \left\{
    \int_{\RR^3}d\bet~
    f^2(\bet)
    \right\}
    \notag\\
    &\triangleq I_1+I_2.
\end{align}

For $I_1$, we have
\begin{align*}
    I_1
    =
    8\pi^2\omega_0^2
    \int_{\RR^3}
    d\eta_1d\eta_2d\eta_3~
    \left|
    F_0\left(
    B_H\phi(\eta_1,\eta_2,\eta_3)
    \right)
    \right|^2
    \left|
    \cG(\eta_1,\eta_2;c_s)
    \right|^4.
\end{align*}

For $I_2$, we note that
\[
\frac{\partial}{\partial\eta_3}
\phi(\eta_1,\eta_2,\eta_3)
=
\frac{2}{c_0}.
\]
For fixed $(\eta_1,\eta_2)$, we introduce the change of variable
\[
\zeta
=
B_H\phi(\eta_1,\eta_2,\eta_3),
\]
which gives
\[
\eta_3
=
\frac{c_0}{2}
\left(
\frac{\zeta}{B_H}
-
\frac{\eta_1^2+\eta_2^2}{c_0|\bz_g|}
\right),
\qquad
\frac{\partial\eta_3}{\partial\zeta}
=
\frac{c_0}{2B_H}.
\]
Consequently,
\begin{align*}
    \int_{\RR^3}d\bet~
    f^2(\bet)
    &=
    -\frac{8\pi^2\omega_0^2c_0}{B_H}
    \int_{\RR^2}d\eta_1d\eta_2
    \int_{\RR}d\zeta~
    e^{i\frac{2\omega_0}{B_H}\zeta}
    F_0^2(\zeta)
    \cG^4(\eta_1,\eta_2;c_s).
\end{align*}
Due to the existence of the rapid phase term, the integral $I_2$ is negligible. Therefore, we have
\begin{align*}
    &\int_{\RR^3}d\bet~
    |\red{\caP_{\rm v}}(t_g,\bet,c_s)|^2
    \\
    &\simeq
    8\pi^2\omega_0^2
    \int_{\RR^3}
    d\eta_1d\eta_2d\eta_3~
    \left|
    F_0\left(
    \frac{B_H}{c_0}
    \left(
    2\eta_3+
    \frac{\eta_1^2+\eta_2^2}{|\bz_g|}
    \right)
    \right)
    \right|^2
    \\
    &\qquad\qquad\times
    \left|
    \cG_{1,1}\left(
    \frac{a_0\omega_0}{c_0|\bz_g|}\eta_1,
    \frac{a_0\omega_0}{c_0|\bz_g|}\eta_2,
    0,
    -\frac{a_0^2\omega_0}{c_0|\bz_g|}
    \left(
    \left(\frac{c_0}{c_s}\right)^2-1
    \right)
    \right)
    \right|^4.
\end{align*}

\subsection{Proof of Proposition \ref{prop: ergodic}}\label{append_4.5}
Here, we provide a proof in the case where $\rho_\mu$ is a stationary
Gaussian process.

The covariance function of the stationary Gaussian process $\rho_\mu$
belongs to $L^1(\mathbb{R}^3)$, and hence $\rho_\mu$ is ergodic.
By Proposition~\ref{prop:stationary}, the leading-order imaging
functional is a measurable translation-equivariant functional of
$\rho_\mu$. Therefore, the corresponding random field
\[
\Delta\bz
\mapsto
\left|
\cI\left(\bz_s(c_s,t_g)+\varepsilon\Delta\bz,c_s\right)
\right|^2
\]
is stationary and ergodic at leading order.

Let $l=l(\varepsilon)$ satisfy
\begin{align}
    l(\varepsilon)\to\infty,
    \qquad
    \varepsilon l(\varepsilon)\to0,
    \qquad
    \text{as }\varepsilon\to0.
\end{align}
Since
\begin{align}
    \sup_{\Delta\bz\in\cS(l(\varepsilon))}
    |\varepsilon\Delta\bz|
    =O(\varepsilon l(\varepsilon))
    \longrightarrow0,
\end{align}
the local stationary representation established in
Proposition~\ref{prop:stationary} remains valid throughout the
averaging region.

More precisely, at leading order, Proposition~\ref{prop:stationary}
gives
\begin{align}
    \cI\left(\bz_s(c_s,t_g)+\varepsilon\Delta\bz,c_s\right)
    \simeq
    C\int_{\RR^3}
    \rho_\mu\left(
    \bz_g+\varepsilon\bp(\Delta\bz,c_s)+\br(\bet)
    \right)
    \red{\caP_{\rm v}}(t_g,\bet,c_s)\,d\bet,
\end{align}
where $C$ is independent of $\Delta\bz$. Here,
\begin{align}
    \bp(\Delta\bz,c_s)
    =
    \left(
    \frac{\Delta\bz^\perp}{v^2},
    \frac{\Delta z^\parallel}{v}
    \right).
\end{align}
Thus, for fixed $c_s$, a translation of $\Delta\bz$ induces a
translation of the spatial argument of the stationary random field
$\rho_\mu$. Hence, the leading-order imaging functional, and therefore
its squared modulus, is stationary and ergodic with respect to the
rescaled variable $\Delta\bz$.

Since $l(\varepsilon)\to\infty$, Birkhoff's ergodic theorem can be
applied over $\cS(l(\varepsilon))$. Moreover,
$\varepsilon l(\varepsilon)\to0$ guarantees that the entire physical
averaging region remains within the local regime. Consequently,
\begin{align}
&
\frac{1}{|\cS(l(\varepsilon))|}
\int_{\cS(l(\varepsilon))}
d\Delta\bz~
\left|
\cI\left(
\bz_s(c_s,t_g)+\varepsilon\Delta\bz,c_s
\right)
\right|^2
\longrightarrow
\EE \left[\left| \cI(\bz_s(c_s,t_g),c_s) \right|^2\right]
\qquad\text{a.s.}
\end{align}
and hence also in probability.

For more general stationary processes, the same conclusion can be
obtained in probability provided that the corresponding spatial
average has vanishing variance.

\subsection{Proof of Corollary \ref{cor: random estimator}}\label{append_4.6}
Here, we offer a qualitative argument. We notice that for fixed $c_s$, the essential support $\cG_{1,1}$ concentrates around $\bet = 0$, and also the essential support of $F_{0}$ concentrate around $0$. The quantity $\EE\left[ |\cI(\bz_s,c_s)|^2\right]$ behaves qualitatively as the function $\left|\cG_{1,1}\bra{0,0,0,\frac{-a_0^2\omega_0}{c_0|\bz_g|}\bra{\bra{\frac{c_0}{c_s}}^2-1}}\right|^4$ up to some multiplicity constant. Then, it is clear that the maximum is achieved at $c_s=c_0$.

\section{Auxiliary lemmas}

\begin{lem}\label{lem: cov}
Let $n(\bx)$ be a random process for $\bx\in\RR^d$, and let
$\Sigma(\bx)$ be its covariance function satisfying
\[
|\Sigma(\bx)|
\lesssim
(1+|\bx|)^{-(d+1)-\eta},
\qquad
\eta>0.
\]
Let $a_1,\ldots,a_d>0$ be small parameters such that $\max_j a_j \ll 1$, 
 and define
\[
n_{\mathbf{a}}(\bx)
=
n\left(
\frac{x_1}{a_1},\ldots,\frac{x_d}{a_d}
\right).
\]
For any deterministic functions
$f\in\mathcal C^1(\RR^d)\cap H^1(\RR^d)$ and
$g\in\mathcal L^2(\RR^d)$, we have
\begin{align}
&\int_{\RR^d\times\RR^d}
d\bx\,d\by~
\operatorname{Cov}
\left(
n_{ \mathbf{a}}(\bx),n_{ \mathbf{a}}(\by)
\right)
f(\bx)\overline{g(\by)}
\notag\\
&\qquad=
\left(\prod_{j=1}^d a_j\right)
\left(
\int_{\RR^d}d\bz~\Sigma(\bz)
\right)
\int_{\RR^d}d\bx~
f(\bx)\overline{g(\bx)}
+
\mathcal O\left(
\left(\prod_{j=1}^d a_j\right)
\max_{1\leq j\leq d}a_j
\right).
\end{align}
\end{lem}

\begin{proof}
By definition,
\[
\operatorname{Cov}
\left(
n_{ \mathbf{a}}(\bx),n_{ \mathbf{a}}(\by)
\right)
=
\Sigma\left(
\frac{x_1-y_1}{a_1},
\ldots,
\frac{x_d-y_d}{a_d}
\right).
\]
Therefore, by the change of variables
\[
z_j=\frac{x_j-y_j}{a_j},
\qquad j=1,\ldots,d,
\]
we obtain
\begin{align*}
&\int_{\RR^d\times\RR^d}
d\bx\,d\by~
\operatorname{Cov}
\left(
n_{ \mathbf{a}}(\bx),n_{ \mathbf{a}}(\by)
\right)
f(\bx)\overline{g(\by)}
\\
&\qquad=
\left(\prod_{j=1}^d a_j\right)
\int_{\RR^d}d\by
\int_{\RR^d}d\bz~
\Sigma(\bz)
f(\by+A\bz)
\overline{g(\by)},
\end{align*}
where
\[
A=\operatorname{diag}(a_1,\ldots,a_d).
\]
Then, we have
\begin{align*}
&\left|
\int_{\RR^d\times\RR^d}
d\bx\,d\by~
\operatorname{Cov}
\left(
n_{ \mathbf{a}}(\bx),n_{ \mathbf{a}}(\by)
\right)
f(\bx)\overline{g(\by)}
\right.
\\
&\qquad\left.
-
\left(\prod_{j=1}^d a_j\right)
\left(
\int_{\RR^d}d\bz~\Sigma(\bz)
\right)
\int_{\RR^d}d\by~
f(\by)\overline{g(\by)}
\right|
\\
&\qquad\lesssim
\left(\prod_{j=1}^d a_j\right)
\int_{\RR^d}d\bz~
|\Sigma(\bz)|
\int_{\RR^d}d\by~
|g(\by)|
\left|
f(\by+A\bz)-f(\by)
\right|.
\end{align*}
By the mean value theorem,
\[
\left|
f(\by+A\bz)-f(\by)
\right|
\leq
|A\bz|
\int_0^1
\left|
\nabla f(\by+tA\bz)
\right|dt.
\]
Since
\[
|A\bz|
\leq
\left(\max_{1\leq j\leq d}a_j\right)|\bz|,
\]
we obtain
\begin{align*}
&\left|
\int_{\RR^d\times\RR^d}
d\bx\,d\by~
\operatorname{Cov}
\left(
n_{ \mathbf{a}}(\bx),n_{ \mathbf{a}}(\by)
\right)
f(\bx)\overline{g(\by)}
\right.
\\
&\qquad\left.
-
\left(\prod_{j=1}^d a_j\right)
\left(
\int_{\RR^d}d\bz~\Sigma(\bz)
\right)
\int_{\RR^d}d\by~
f(\by)\overline{g(\by)}
\right|
\\
&\qquad\lesssim
\left(\prod_{j=1}^d a_j\right)
\left(\max_{1\leq j\leq d}a_j\right)
\left(
\int_{\RR^d}d\bz~
|\bz|\,|\Sigma(\bz)|
\right)
\left(
\int_{\RR^d}d\by~
|\nabla f(\by)||g(\by)|
\right).
\end{align*}
The assumed decay of $\Sigma$ implies
\[
\int_{\RR^d}
|\bz|\,|\Sigma(\bz)|\,d\bz<\infty,
\]
which proves the result.
\end{proof}

\begin{lem}\label{lem: integral}
    We present an integral which is frequently used in the calculation of the point spread function: 
    \begin{align}
        \int_\RR dx ~e^{i(-Ax^2+Bx)} = \sqrt{\frac{\pi}{A}}e^{-i\frac{\pi}{4}}e^{i\frac{B^2}{4A}}, \blue{\quad A>0}.
    \end{align}
\end{lem}
\begin{proof}
    \begin{align}
        \int_\RR dx~ e^{i(-Ax^2+Bx)} = \int_\RR dx~ e^{-iA(x-\frac{B}{2A})^2}\cdot e^{i\frac{B^2}{4A}}= \sqrt{\frac{\pi}{A}}e^{-i\frac{\pi}{4}}e^{i\frac{B^2}{4A}},
    \end{align}
    where we use the identity that 
    \begin{align*}
        \int_\RR ds~e^{-i\frac{s^2}{2}}  = \sqrt{2\pi}e^{-i\frac{\pi}{4}}.
    \end{align*}
\end{proof}

\bibliographystyle{IEEEtran}
\bibliography{reference}

@book{Garnier_Papanicolaou_2016, place={Cambridge}, title={Passive Imaging with Ambient Noise}, publisher={Cambridge University Press}, author={Garnier, Josselin and Papanicolaou, George}, year={2016}}

@article{Josselin_Garnier_2025_Inverse_Problems_and_Imaging,
  title = {Probing the speckle to estimate the effective speed of sound, a first step towards quantitative ultrasound imaging},
  journal = {Inverse Problems and Imaging},
  year = {2025},
  pages = {},
  issn = {1930-8337},
  doi = {10.3934/ipi.2026001},
  url = {https://www.aimsciences.org/article/id/692581ef1511553245b010e2},
  author = {Josselin Garnier and Laure Giovangigli and Quentin Goepfert and Pierre Millien}
}

@article{de2009semiclassical,
  title={Semiclassical analysis and passive imaging},
  author={De Verdi{\`e}re, Yves Colin},
  journal={Nonlinearity},
  volume={22},
  number={6},
  pages={R45},
  year={2009},
  publisher={IOP Publishing}
}

@article{Bardos_2008,
doi = {10.1088/0266-5611/24/1/015011},
url = {https://doi.org/10.1088/0266-5611/24/1/015011},
year = {2008},
month = {jan},
publisher = {},
volume = {24},
number = {1},
pages = {015011},
author = {Bardos, Claude and Garnier, Josselin and Papanicolaou, George},
title = {Identification of Green's functions singularities by cross correlation of noisy signals},
journal = {Inverse Problems}
}

@article{doi:10.1137/080723454,
author = {Garnier, Josselin and Papanicolaou, George},
title = {Passive Sensor Imaging Using Cross Correlations of Noisy Signals in a Scattering Medium},
journal = {SIAM Journal on Imaging Sciences},
volume = {2},
number = {2},
pages = {396-437},
year = {2009},
doi = {10.1137/080723454},

URL = { 
    
        https://doi.org/10.1137/080723454
    
    

},
eprint = { 
    
        https://doi.org/10.1137/080723454
    
    

}
}

@article{garnier2010resolution,
  title={Resolution analysis for imaging with noise},
  author={Garnier, Josselin and Papanicolaou, George},
  journal={Inverse Problems},
  volume={26},
  number={7},
  pages={074001},
  year={2010},
  publisher={IOP Publishing}
}

@article{HELIN2018132,
title = {Correlation based passive imaging with a white noise source},
journal = {Journal de Mathématiques Pures et Appliquées},
volume = {116},
pages = {132-160},
year = {2018},
issn = {0021-7824},
doi = {https://doi.org/10.1016/j.matpur.2018.05.001},
url = {https://www.sciencedirect.com/science/article/pii/S0021782418300710},
author = {T. Helin and M. Lassas and L. Oksanen and T. Saksala}
}

@article{muller2024quantitative,
  title={Quantitative passive imaging by iterative holography: the example of helioseismic holography},
  author={M{\"u}ller, Bj{\"o}rn and Hohage, Thorsten and Fournier, Damien and Gizon, Laurent},
  journal={Inverse Problems},
  volume={40},
  number={4},
  pages={045016},
  year={2024},
  publisher={IOP Publishing}
}

@article{lambert2022ultrasound,
  title={Ultrasound matrix imaging—Part I: The focused reflection matrix, the F-factor and the role of multiple scattering},
  author={Lambert, William and Robin, Justine and Cobus, Laura A and Fink, Mathias and Aubry, Alexandre},
  journal={IEEE Transactions on Medical Imaging},
  volume={41},
  number={12},
  pages={3907--3920},
  year={2022},
  publisher={IEEE}
}

@ARTICLE{9858891,
  author={Lambert, William and Cobus, Laura A. and Robin, Justine and Fink, Mathias and Aubry, Alexandre},
  journal={IEEE Transactions on Medical Imaging}, 
  title={Ultrasound Matrix Imaging—Part II: The Distortion Matrix for Aberration Correction Over Multiple Isoplanatic Patches}, 
  year={2022},
  volume={41},
  number={12},
  pages={3921-3938},
  doi={10.1109/TMI.2022.3199483}}

@article{bureau2024reflection,
  title={Reflection matrix imaging for wave velocity tomography},
  author={Bureau, Flavien and Giraudat, Elsa and Ber, Arthur Le and Lambert, William and Carmier, Louis and Guibal, Aymeric and Fink, Mathias and Aubry, Alexandre},
  journal={arXiv preprint arXiv:2409.13901},
  year={2024}
}

@article{giraudat2024matrix,
  title={Matrix imaging as a tool for high-resolution monitoring of deep volcanic plumbing systems with seismic noise},
  author={Giraudat, Elsa and Burtin, Arnaud and Le Ber, Arthur and Fink, Mathias and Komorowski, Jean-Christophe and Aubry, Alexandre},
  journal={Communications Earth \& Environment},
  volume={5},
  number={1},
  pages={509},
  year={2024},
  publisher={Nature Publishing Group UK London}
}

@article{agaltsov2020global,
  title={Global uniqueness in a passive inverse problem of helioseismology},
  author={Agaltsov, AD and Hohage, Thorsten and Novikov, RG},
  journal={Inverse Problems},
  volume={36},
  number={5},
  pages={055004},
  year={2020},
  publisher={IOP Publishing}
}

@article{https://doi.org/10.1111/j.1365-2478.2007.00684.x,
author = {Gouédard, P. and Stehly, L. and Brenguier, F. and Campillo, M. and Colin de Verdière, Y. and Larose, E. and Margerin, L. and Roux, P. and Sánchez-Sesma, F. J. and Shapiro, N. M. and Weaver, R. L.},
title = {Cross-correlation of random fields: mathematical approach and applications},
journal = {Geophysical Prospecting},
volume = {56},
number = {3},
pages = {375-393},
doi = {https://doi.org/10.1111/j.1365-2478.2007.00684.x},
url = {https://onlinelibrary.wiley.com/doi/abs/10.1111/j.1365-2478.2007.00684.x},
eprint = {https://onlinelibrary.wiley.com/doi/pdf/10.1111/j.1365-2478.2007.00684.x},
year = {2008}
}

@article{harmankaya2013locating,
  title={Locating near-surface scatterers using non-physical scattered waves resulting from seismic interferometry},
  author={Harmankaya, Utku and Kaslilar, Ayse and Thorbecke, J and Wapenaar, K and Draganov, D},
  journal={Journal of Applied Geophysics},
  volume={91},
  pages={66--81},
  year={2013},
  publisher={Elsevier}
}

@article{kaslilar2013estimating,
  title={Estimating the location of a tunnel using correlation and inversion of Rayleigh wave scattering},
  author={Kaslilar, Ayse and Harmankaya, Utku and Wapenaar, K and Draganov, D},
  journal={Geophysical Research Letters},
  volume={40},
  number={23},
  pages={6084--6088},
  year={2013},
  publisher={Wiley Online Library}
}

@article{konstantaki2013imaging,
  title={Imaging scatterers in landfills using seismic interferometry},
  author={Konstantaki, Laura Amalia and Draganov, Deyan and Heimovaara, Timo and Ghose, Ranajit},
  journal={Geophysics},
  volume={78},
  number={6},
  pages={EN107--EN116},
  year={2013},
  publisher={Society of Exploration Geophysicists}
}

@article{bakulin2006virtual,
  title={The virtual source method: Theory and case study},
  author={Bakulin, Andrey and Calvert, Rodney},
  journal={Geophysics},
  volume={71},
  number={4},
  pages={SI139--SI150},
  year={2006},
  publisher={Society of Exploration Geophysicists}
}

@article{wapenaar2017virtual,
  title={Virtual sources and their responses, Part I: time-reversal acoustics and seismic interferometry},
  author={Wapenaar, Kees and Thorbecke, Jan},
  journal={Geophysical Prospecting},
  volume={65},
  number={6},
  pages={1411--1429},
  year={2017},
  publisher={European Association of Geoscientists \& Engineers}
}

@article{garnier2014role,
  title={Role of scattering in virtual source array imaging},
  author={Garnier, Josselin and Papanicolaou, George},
  journal={SIAM Journal on Imaging Sciences},
  volume={7},
  number={2},
  pages={1210--1236},
  year={2014},
  publisher={SIAM}
}

@article{https://doi.org/10.1111/1365-2478.12495,
author = {Wapenaar, Kees and Thorbecke, Jan and van der Neut, Joost and Slob, Evert and Snieder, Roel},
title = {Review paper: Virtual sources and their responses, Part II: data-driven single-sided focusing},
journal = {Geophysical Prospecting},
volume = {65},
number = {6},
pages = {1430-1451},
doi = {https://doi.org/10.1111/1365-2478.12495},
url = {https://onlinelibrary.wiley.com/doi/abs/10.1111/1365-2478.12495},
eprint = {https://onlinelibrary.wiley.com/doi/pdf/10.1111/1365-2478.12495},
year = {2017}
}

@article{shapiro2005high,
  title={High-resolution surface-wave tomography from ambient seismic noise},
  author={Shapiro, Nikolai M and Campillo, Michel and Stehly, Laurent and Ritzwoller, Michael H},
  journal={Science},
  volume={307},
  number={5715},
  pages={1615--1618},
  year={2005},
  publisher={American Association for the Advancement of Science}
}

@article{lin2009eikonal,
  title={Eikonal tomography: surface wave tomography by phase front tracking across a regional broad-band seismic array},
  author={Lin, Fan-Chi and Ritzwoller, Michael H and Snieder, Roel},
  journal={Geophysical Journal International},
  volume={177},
  number={3},
  pages={1091--1110},
  year={2009},
  publisher={Blackwell Publishing Ltd Oxford, UK}
}

@article{ritzwoller2011ambient,
  title={Ambient noise tomography with a large seismic array},
  author={Ritzwoller, Michael H and Lin, Fan-Chi and Shen, Weisen},
  journal={Comptes Rendus Geoscience},
  volume={343},
  number={8-9},
  pages={558--570},
  year={2011},
  publisher={Elsevier}
}

@article{sager2018towards,
  title={Towards full waveform ambient noise inversion},
  author={Sager, Korbinian and Ermert, Laura and Boehm, Christian and Fichtner, Andreas},
  journal={Geophysical Journal International},
  volume={212},
  number={1},
  pages={566--590},
  year={2018},
  publisher={Oxford University Press}
}

@incollection{fichtner_tsai_2019,
  author    = {Andreas Fichtner and Victor C. Tsai},
  title     = {Theoretical foundations of noise interferometry},
  booktitle = {Seismic Ambient Noise},
  editor    = {Nori Nakata and Lucia Gualtieri and Andreas Fichtner},
  publisher = {Cambridge University Press},
  year      = {2019},
  chapter   = {4},
  pages     = {109--143}
}

@article{weaver2001ultrasonics,
  title={Ultrasonics without a source: Thermal fluctuation correlations at MHz frequencies},
  author={Weaver, Richard L and Lobkis, Oleg I},
  journal={Physical Review Letters},
  volume={87},
  number={13},
  pages={134301},
  year={2001},
  publisher={APS}
}

@article{duvall1993time,
  title={Time--distance helioseismology},
  author={Duvall Jr, Thomas L and Jeffferies, SM and Harvey, JW and Pomerantz, MA},
  journal={Nature},
  volume={362},
  number={6419},
  pages={430--432},
  year={1993},
  publisher={Nature Publishing Group UK London}
}

@article{sabra2006extracting,
  title={Extracting coherent coda arrivals from cross-correlations of long period seismic waves during the Mount St. Helens 2004 eruption},
  author={Sabra, Karim G and Roux, Philippe and Gerstoft, Peter and Kuperman, WA and Fehler, Michael C},
  journal={Geophysical research letters},
  volume={33},
  number={6},
  year={2006},
  publisher={Wiley Online Library}
}

@article{curtis2006seismic,
  title={Seismic interferometry—Turning noise into signal},
  author={Curtis, Andrew and Gerstoft, Peter and Sato, Haruo and Snieder, Roel and Wapenaar, Kees},
  journal={The Leading Edge},
  volume={25},
  number={9},
  pages={1082--1092},
  year={2006},
  publisher={Society of Exploration Geophysicists}
}

@article{larose2006correlation,
  title={Correlation of random wavefields: An interdisciplinary review},
  author={Larose, Eric and Margerin, Ludovic and Derode, Arnaud and van Tiggelen, Bart and Campillo, Michel and Shapiro, Nikolai and Paul, Anne and Stehly, Laurent and Tanter, Mickael},
  journal={Geophysics},
  volume={71},
  number={4},
  pages={SI11--SI21},
  year={2006},
  publisher={Society of Exploration Geophysicists}
}

@article{lobkis2001emergence,
  title={On the emergence of the Green’s function in the correlations of a diffuse field},
  author={Lobkis, Oleg I and Weaver, Richard L},
  journal={The Journal of the Acoustical Society of America},
  volume={110},
  number={6},
  pages={3011--3017},
  year={2001},
  publisher={Acoustical Society of America}
}

@article{PhysRevE.69.046610,
  title = {Extracting the Green's function from the correlation of coda waves: A derivation based on stationary phase},
  author = {Snieder, Roel},
  journal = {Phys. Rev. E},
  volume = {69},
  issue = {4},
  pages = {046610},
  numpages = {8},
  year = {2004},
  month = {Apr},
  publisher = {American Physical Society},
  doi = {10.1103/PhysRevE.69.046610},
  url = {https://link.aps.org/doi/10.1103/PhysRevE.69.046610}
}

@article{Garnier_2012,
doi = {10.1088/0266-5611/28/7/075002},
url = {https://doi.org/10.1088/0266-5611/28/7/075002},
year = {2012},
month = {jun},
publisher = {IOP Publishing},
volume = {28},
number = {7},
pages = {075002},
author = {Garnier, Josselin and Papanicolaou, George},
title = {Correlation-based virtual source imaging in strongly scattering random media},
journal = {Inverse Problems}
}

@article{Garnier_2009,
doi = {10.1088/0266-5611/25/4/045005},
url = {https://doi.org/10.1088/0266-5611/25/4/045005},
year = {2009},
month = {feb},
publisher = {},
volume = {25},
number = {4},
pages = {045005},
author = {Garnier, Josselin and Sølna, Knut},
title = {Background velocity estimation with cross correlations of incoherent waves in the parabolic scaling},
journal = {Inverse Problems}
}

@book{BornWolf1993,
  author    = {Max Born and Emil Wolf},
  title     = {Principles of Optics},
  edition   = {6},
  publisher = {Pergamon Press},
  address   = {Oxford, United Kingdom},
  year      = {1993}
}

\end{document}